%% file: median-handcuff-homomorphisms-4.tex
\documentclass[a4paper]{amsart}
\usepackage{graphicx,xcolor,hyperref}
\usepackage{a4wide}

\DeclareMathOperator{\ar}{ar} 
\newcommand{\ignore}[1]{}
\DeclareMathOperator{\low}{low}
\DeclareMathOperator{\res}{res}
\DeclareMathOperator{\high}{high}

\newtheorem{theorem}{Theorem}[section]
\newtheorem{definition}[theorem]{Definition}
\newtheorem{lemma}[theorem]{Lemma}
\newtheorem{corollary}[theorem]{Corollary}

\newcommand\red[1]{#1}
\newcommand\blue[1]{#1}

\DeclareMathOperator{\median}{median}
\DeclareMathOperator{\Csp}{CSP}
\newtheorem{expl}[theorem]{Example}
\newtheorem{proposition}[theorem]{Proposition}
\title[Median-Closed Semilinear CSPs]{A polynomial-time algorithm for \\
median-closed semilinear constraints}
\author{Manuel Bodirsky}
\author{Marcello Mamino}
\date{\today}

\begin{document}

\maketitle


\begin{abstract}
A subset of ${\mathbb Q}^k$ is called \emph{semilinear} (or \emph{piecewise linear})
if it is a Boolean combination of linear half-spaces. 
We study the computational complexity of the constraint satisfaction problem (CSP) over the rationals when all the constraints are semilinear.
When the sets are convex the CSP is polynomial-time equivalent to linear programming. A semilinear relation is convex
if and only if it is preserved by taking \emph{averages}. 
Our main result is a polynomial-time algorithm 
 for the CSP of semilinear constraints that are preserved by applying \emph{medians}. 
 We also prove that this class is maximally tractable in the sense that any larger class of 
 semilinear relations has an NP-hard CSP. 
To illustrate, our class contains all relations that can be expressed by linear inequalities with at most two variables (so-called TVPI constraints),
but it also contains many non-convex relations,
for example constraints 
 of the form
$x \in S$ for arbitrary finite $S \subseteq {\mathbb Q}$, or more generally disjunctive constraints of the form
$x \circ_1 c \vee y \circ_2 d$ for $\circ_1,\circ_2 \in \{\leq,\geq,<,>\}$ and $c,d \in {\mathbb Q}$. 
\end{abstract}

\section{Introduction}
The problem of deciding whether a set of linear inequalities has a solution over the rational numbers is one of the most important computational problems, both in theory and practice of programming (see, e.g., \cite{Smale}). The problem can be solved in polynomial time, e.g., by the ellipsoid method~\cite{Khachiyan}. 
It is known that the problem
remains in P even if some of the input inequalities
are strict, or if the input also contains inequalities
of the form $a_1 x_1 + \cdots + a_k x_k \neq a_0$
where $x_1,\dots,x_k$ are variables and  $a_0,a_1,\dots,a_k$ are rational numbers. 
On the other hand, if we additionally
allow constraints of the form $x \in \{0,1\}$ in the input, i.e., if we can require that some of the variables either take value $0$ or value $1$, the
resulting computational problem is easily seen to be NP-complete. 

One might ask which constraint relations can be additionally allowed in the input so that the respective satisfiability problem remains in P. This question has been solved completely for sets of \emph{semilinear}
(also known as \emph{piecewise linear}) relations. 
A subset of ${\mathbb Q}^k$ is semilinear 
if it is a Boolean combination of closed linear half-spaces. If each of the additional semilinear relations $R \subseteq {\mathbb Q}^k$ is \emph{essentially convex}, i.e., if for all $x,y \in R$ 
there are only finitely many points on the line segment between $x$ and $y$ that are \emph{not} contained in $R$, then the constraint satisfaction problem (CSP) can be solved in polynomial time, and otherwise it is NP-complete~\cite{Essentially-convex}. 
Therefore, the class of essentially convex
relations is called \emph{maximally tractable (within the class of semilinear relations)} in the sense that adding any semilinear relation to it which is not essentially convex has an NP-hard CSP. 

A semilinear relation $R \subseteq {\mathbb Q}^k$ is convex if and only if it is preserved by applying component-wise the average operation, $(x,y)\mapsto \frac{x+y}{2}$. Such preservation conditions (also called \emph{polymorphisms}; see Section~\ref{sect:majority}) have been crucial in solving the famous Feder-Vardi dichotomy conjecture for the complexity of finite-domain CSPs~\cite{FederVardi,BulatovFVConjecture,ZhukFVConjecture}. 
In this paper we identify a new preservation condition for semilinear relations that leads to polynomial-time tractability of the CSP, namely being preserved by componentwise median. 
Examples of semilinear relations preserved by
median are so-called TVPI constraints (for \emph{two variables per inequality}; they are also preserved by average). 
Our class is maximally tractable: adding any  relation that is not median-closed leads to an NP-hard CSP 
(Theorem~\ref{thm:maximality}). 
It contains all
constraints of the form $x \circ_1 c_1 \vee y \circ_2 c_2$ for 
$\circ_1,\circ_2 \in \{\leq,\geq,<,>\}$ and 
arbitrary rationals $c_1,c_2$, among many other non essentially-convex relations. 
Note that with these constraints we can also express the constraint $x \in \{0,1\}$ as
$$x \geq 0 \wedge x \leq 1 \wedge (x \leq 0 \vee x \geq 1)$$ and more generally we can express the constraint $x \in S$ for any finite $S \subseteq {\mathbb Q}$. 
These constraints are practically relevant 
since in many applications we are looking for solutions where some of the variables may take values from ${\mathbb Q}$ while other variables range only over a finite set.

Over a finite domain, CSPs with a \emph{majority polymorphism} (such as the median) 
are known to be solvable in polynomial time by establishing local consistency~\cite{FederVardi,CCC}. This result does not extend to our setting, since it is unclear how to establish $(2,3)$-consistency for TVPI constraints in polynomial time.
More recently, it has been shown that if a finite-domain CSP can be solved by establishing local consistency, then already (a restricted version of) singleton arc-consistency solves the CSP~\cite{Kozik16}. It is not clear how to
extend the results to infinite domains;
moreover, it is not clear how to establish singleton arc-consistency of median-closed semilinear constraints in polynomial time.

For our polynomial-time algorithm we combine
universal-algebraic ideas from~\cite{BP,CCC,FederVardi} and algorithmic ideas from~\cite{Shostak,AspvallShiloach,HochbaumNaor}.
First we provide an explicit characterisation of those semilinear relations that are preserved by
median, which in particular allows us to reduce 
general median-closed relations to binary ones. 
Then, as in Hochbaum-Naor~\cite{HochbaumNaor}, 
we use bound propagation together with variable elimination \`a la Fourier-Motzkin. In order to deal with the disjunctive constraints in the bound propagation, 
we have to generalise the result of Shostak about unsatisfiability of TVPI instances~\cite{Shostak}. Our algorithm is \emph{strongly polynomial}, unlike the known algorithms for the linear program feasibility problem.
The existence 
of a strongly polynomial algorithm for linear programming is a famous open problem~\cite{Smale} (Smale's 9th problem), 
whereas for TVPI constraints, strongly polynomial algorithms were known before~\cite{AspvallShiloach,Megiddo:1983:TGP,HochbaumNaor}.


We mention that the class of median-closed
constraints also appears in the study of 
\emph{valued constraint satisfaction problems} 
for cost functions that are submodular; for the definitions of the concepts that appear in this paragraph, see~\cite{KolmogorovThapperZivny15,GenVCSP15}.
A function is \emph{submodular} with respect to some linearly ordered domain if and only if it is preserved by a certain binary symmetric fractional polymorphism,
namely the fractional polymorphism that equals
$\min$ and $\max$ with probability $0.5$ each;
hence, the support of submodular cost functions must be both $\min$- and $\max$-closed. 
The valued constraint satisfaction problem for a large class of submodular semilinear cost functions was shown to be in P~\cite{BodirskyMaminoViola};
however, for the class of \emph{all} submodular semilinear cost functions, the complexity has
not yet been classified. 
Since the
median operation can be expressed as
$$\min(\max(x,y),\max(y,z),\max(x,z))$$
the support of a submodular function is also $\median$-closed. 
Hence, our result also implies that the `crisp part'
of submodular semilinear valued constraint satisfaction problems (i.e., the feasibility problem for these optimisation problems) can be solved in polynomial time.

\medskip 
{\bf Outline.} 
In Section~\ref{sect:csps} we formally introduce CSPs and polymorphisms. 
In Section~\ref{sect:syntax} we provide an explicit description of the semilinear relations preserved by median in terms of syntactically restricted quantifier-free formulas with rational parameters over the structure $({\mathbb Q};+,\leq)$. 
Section~\ref{sect:shostak} presents 
a generalisation of Shostak's theorem which characterises unsatisfiability not only for TVPI constraints, but more generally for instances of median-closed semilinear constraints. 
We then present in Section~\ref{sect:alg} the algorithm for general median-closed constraints. 
The maximal tractability of median-closed semilinear constraints is treated in  Section~\ref{sect:maximality}. 
In Section~\ref{sect:open} we mention open research problems that are relevant for 
the more ambitious research goal of classifying the computational complexity of \emph{all} semilinear constraints languages.

\section{Constraint Satisfaction Problems}
\label{sect:csps}
Let $\tau$ be a relational signature,
i.e., a set of relation symbols $R$, 
each equipped with an arity $\ar(R) \in {\mathbb N}$. 
A \emph{$\tau$-structure} $\Gamma$
consists of a set $D$ (the \emph{domain}) together with a relation
$R^{\Gamma} \subseteq D^{\ar(R)}$ 
for each $R \in \tau$. 
The \emph{constraint satisfaction problem}
for a structure $\Gamma$ with finite relational signature $\tau$ is the computational problem of deciding for a given
conjunction $\Phi$ of atomic $\tau$-formulas
whether $\Phi$ is satisfiable over $\Gamma$. 

\subsection{Polymorphisms}
A function $f \colon D^k \to D$ of arity $k \in {\mathbb N}$ \emph{preserves}
a relation $R \subseteq D^m$ if 
$$(f(a_{1,1},\dots,a_{k,1}),\dots,f(a_{1,m},\dots,a_{k,m})) \in R$$
for all
$(a_{1,1},\dots,a_{1,m}),\dots,(a_{k,1},\dots,a_{k,m}) \in R$. 
A \emph{polymorphism} of a relational structure $\Gamma$ with domain $D$ is a function 
$f \colon D^k \to D$, for some $k \in {\mathbb N}$,  that preserves all relations in $\Gamma$. 
For structures $\Gamma$ with a finite domain, 
it is known that the computational complexity
of $\Csp(\Gamma)$ only depends on 
the set of polymorphisms of $\Gamma$ (see~\cite{JBK}). 

\subsection{Majority Polymorphisms}
\label{sect:majority} 
An operation $f \colon D^3 \to D$ satisfying
$$f(x,x,y) = f(x,y,x) = f(y,x,x) = x$$
is called a \emph{majority operation}. 
An example of a majority operation is the \emph{median-operation}. The definition of the median-operation requires $D$ to be linearly ordered. Let $x,y,z \in D$. Choose $u,v,w \in D$ such that $\{u,v,w\} = \{x,y,z\}$ and $u \leq v \leq w$. Then $\median$ is defined by 
$$\median(x,y,z) := v \, .$$
It is well-known (see, e.g.,~\cite{Rendezvous}) that a relation over the Boolean domain $\{0,1\}$, with $0<1$, is preserved by
the median operation if and only if it can be defined by 2-SAT formula (i.e., by a Boolean formula in conjunctive normal form where each clause has at most two variables). 
An example of a subset of ${\mathbb Q}^2$ that is preserved by the median operation can be found in Figure~\ref{fig:median}. 

A relation $R \subseteq D^m$ 
is called \emph{2-decomposable} 
if a tuple $(t_1,\dots,t_m) \in D^m$
 is contained in $R$ if and only if
for all $i,j \leq m$ there exists a tuple
$(s_1,\dots,s_m) \in R$ with $s_i=t_i$ and $s_j=t_j$. 
In other words, $R$ can be expressed as a conjunction of binary projections of $R$. 

\begin{theorem}[of \cite{CCC,BP}]\label{thm:decomp}
Let $\Gamma$ be a structure with domain $D$ and a majority polymorphism. 
Then every relation in $\Gamma$ is 2-decomposable. 
\end{theorem}
Note that Theorem~3.5 in~\cite{CCC} (the implication $(1) \Rightarrow (2)$) states this result only for finite domains, but Section 4.4 explicitly treats 
the case of infinite domains and mentions that the proof of Theorem 3.5 remains unchanged after a modification
of the statement that does not involve item (1) and (2). 


\section{Semilinear Relations}
\label{sect:syntax}
A relation $R \subseteq {\mathbb Q}^k$ of arity $k \in {\mathbb N}$ is called
\emph{semilinear} if there exists a first-order formula $\phi(x_1,\dots,x_k)$ that \emph{defines} it over $({\mathbb Q};+,1,\leq)$, i.e., 
we have $(u_1,\dots,u_k) \in R$ if and only if $\phi(u_1,\dots,u_k)$ holds in $({\mathbb Q};+,1,\leq)$. 
Ferrante and Rackoff~\cite{FerranteRackoff}
showed that the structure $({\mathbb Q};+,1,\leq)$
has quantifier elimination, and consequently that
every semilinear relation is a Boolean combination
of closed half-spaces.

A formula of the form
$x \circ d$ with $d \in {\mathbb Q}$ is called an \emph{upper bound (on $x$)} if $\circ \in \{\leq,<\}$,
and a \emph{lower bound (on $x$)} if $\circ \in \{\geq,>\}$. A \emph{bound (on $x$)} is either
a lower or an upper bound. 
Bounds of the form $x \leq d$ or $x \geq d$ are called \emph{weak bounds}, and bounds of the form $x < d$ or $x > d$ are called \emph{strict bounds}. 
We additionally allow that $d = +\infty$ or that  
$d = -\infty$; 
the bounds $x \leq +\infty$ and $x \geq -\infty$ are satisfied by all $x \in {\mathbb Q}$,
and the bounds $x \geq +\infty$ and $x \leq -\infty$ are satisfied by no $x \in {\mathbb Q}$.

\begin{figure}
  \begin{center}
  \includegraphics[scale=0.3]{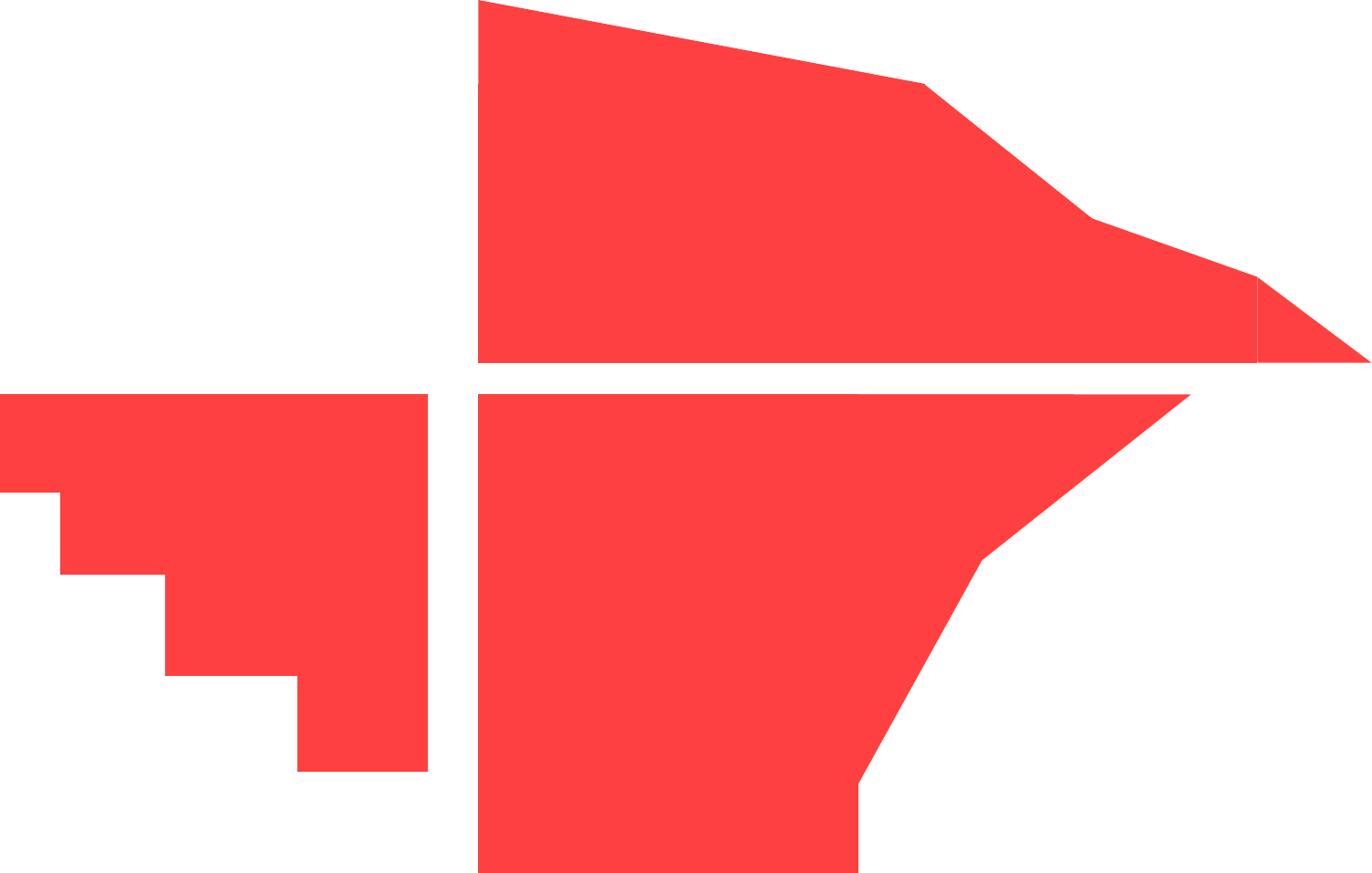}
  \end{center}
\caption{An example of a median-closed subset of ${\mathbb Q}^2$.}
  \label{fig:median}
\end{figure}

Let $a,b,c \in {\mathbb Q}$ be rational numbers. 
A \emph{two-variable weak linear inequality} (short: \emph{TVPI}) on the variables $x$ and $y$ is a
formula 
of the form $ax + by \leq c$, and 
a \emph{two-variable strict linear inequality} (on the variables $x$ and $y$) is a formula 
of the form $ax + by < c$. 
A \emph{two-variable linear inequality} is either
a weak or a strict two-variable linear inequality. 
We make the convention that in two-variable
linear inequalities we have $a \neq 0$ since otherwise we can replace it by an equivalent bound on $y$;
likewise, we assume that 
$b \neq 0$. If both $a$ and $b$ are $0$,
then the inequality is either equivalent
to $\top$ (true) or $\perp$ (false). 
Also in the special case where both variables 
are equal we have that the formula 
$ax+bx \circ c$, for $\circ \in \{<,>,\leq,\geq\}$, 
is equivalent to a bound on $x$. 
The formula $ax+bx \leq c$ is equivalent to 
$x \geq d$ or $x \leq d$ with $d \in \{-\infty,+\infty\}$
if and only if $a+b = 0$. So we also make
the convention that in two-variable linear inequalities
the two variables are not the same, since otherwise we can replace it by an equivalent bound. 

\begin{definition}\label{def:bend}
A \emph{bend} is a formula of the form 
$$x \circ_1 d_1 \; \vee \; (a_1x+a_2y) \circ c \; 
\vee \; y \circ_2 d_2$$
where $\circ \in \{\leq,<\}$, 
$\circ_1, \circ_2 \in \{\leq,\geq,<,>\}$,
$a_1,a_2 \in {\mathbb Q} \setminus \{0\}$, 
and $c,d_1,d_2 \in {\mathbb Q} \cup \{-\infty,+\infty\}$
are \blue{such that $\circ_i \in \{\leq,<\}$ if and only if $a_i>0$, for $i=1$ and $i=2$.}
\end{definition}

Note that by choosing $c = -\infty$ 
a bend can
be used to express any disjunction 
$x \circ_1 d_1 \vee x \circ_2 d_2$ of two bounds. 
Also note that by choosing $d_1,d_2 \in \{-\infty,+\infty\}$
one can also use bends to express arbitrary two-variable linear inequalities, bounds, $\top$ (true), and $\perp$ (false). We therefore view all of these formulas as bends, too. In particular, when we remove a disjunct from a bend we again obtain a bend. A formula is called \emph{bijunctive} if it is a conjunction of bends.



\begin{figure}
  \begin{center}
  \includegraphics[scale=0.3]{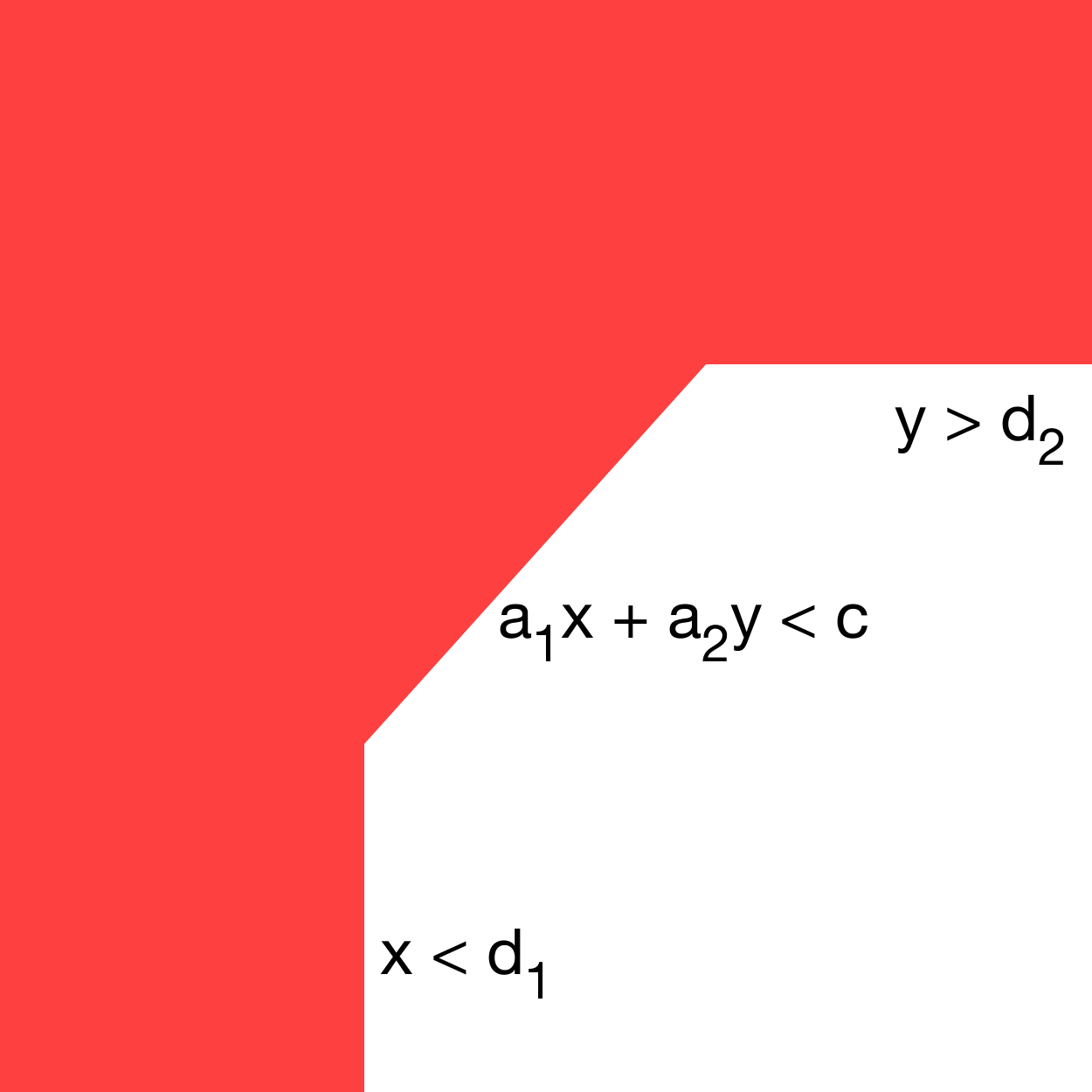}
  \end{center}
\caption{Illustration of the bend $x < d_1 \vee ax + by < c \vee y > d_2$.}
  \label{fig:bend}
\end{figure}
\ignore{
\begin{itemize}
\item $a,b > 0$ and 
$\circ_1,\circ_3 \in \{\leq,<\}$;
\item $a > 0$, $b < 0$, 
and $\circ_1 \in \{\leq,<\}$, $\circ_3 \in \{\geq,>\}$; 
\item $a < 0$, $b > 0$, 
and $\circ_1 \in \{\geq,>\}$, $\circ_3 \in \{\leq,<\}$; 
\item $a,b < 0$ and $\circ_1,\circ_3 \in \{\geq,>\}$. 
\end{itemize}
}

\begin{theorem}\label{thm:syntax} 
Let $R \subseteq {\mathbb Q}^k$ be a semilinear relation.
Then the following are equivalent. 
\begin{enumerate}
\item $R$ is preserved by the median operation.  
\item $R$ can be defined by a conjunction of binary relations $S$ with the property that for each $(u,v) \in {\mathbb Q}^2 \setminus S$
at least one of the following 
sets is contained in 
${\mathbb Q}^2 \setminus S$, too:
\begin{align*}
U_1 & := \{(u',v') \mid u' \geq u, v' 
\geq v\}, & U_2 & := \{(u',v') \mid u' \geq u, v' \leq v\}, \\
U_3 & := \{(u',v') \mid u' \leq u, v' \geq v\}, &
U_4 & := \{(u',v') \mid u' \leq u, v' \leq v\} \\
U_5 & := \{(u',v) \mid u' \in {\mathbb Q}\},
& U_6 & := \{(u,v') \mid v' \in {\mathbb Q}\}. 
\end{align*}
\item $R$ can be defined by a bijunctive formula. 
\end{enumerate}
\end{theorem}

\begin{proof}
\fbox{$(1)$ implies $(2)$.} 
Since the median is a majority operation,  
Theorem~\ref{thm:decomp} implies that $R$ can
be defined by a conjunction of binary projections of $R$; clearly, these projections are again median-closed. So it suffices to show (2) for binary relations $R$. 
Let $(u,v) \in {\mathbb Q}^2 \setminus R$ and suppose for contradiction 
that each of the sets $U_i$, for $i 
\in \{1,\dots,6\}$, contains a point $p_i$ from $R$. 
Suppose that $p_5 = (u,v')$ for $v' < v$
and that $p_6 = (u',v)$ for $u' < u$. 
Then $\median(p_5,p_6,p_1) = (u,v) \in S$
since $S$ is preserved by $\median$,
in contradiction to the assumption that $(u,v) \notin R$. The cases where $p_5 = (u,v')$ for $v' > v$
or $p_6 = (u',v)$ for $u' > u$ can be treated similarly. 

\medskip 
\fbox{$(2)$ implies $(3)$.} Let $R \subseteq {\mathbb Q}^2$ be a relation that satisfies the property from $(2)$. Then ${\mathbb Q}^2 \setminus R$
is the union of the following six semilinear sets: 
\begin{align*}
V_1 & := \{(u,v) \mid (u',v') \notin R \text{ for all } u' \geq u,v' \geq v\} \\ 
V_2 & := \{(u,v) \mid (u',v') \notin R \text{ for all } u' \geq u,v' \leq v\} \\ 
V_3 & := \{(u,v) \mid (u',v') \notin R \text{ for all } u' \leq u,v' \geq v\} \\ 
V_4 & := \{(u,v) \mid (u',v') \notin R \text{ for all } u' \leq u,v' \leq v\} \\ 
V_5 & := \{(u,v) \mid (u,v') \notin R \text{ for all } v' \in {\mathbb Q} \}  \\
V_6 & := \{(u,v) \mid (u',v) \notin R \text{ for all } u' \in {\mathbb Q} \}  
\end{align*}
So it suffices to describe the complements of these six sets using bijunctive formulas. 
For $V_5$ and $V_6$ this can be
done by conjunctions of formulas of the form $x \circ_1 d_1 \vee x \circ_2 d_2$ for $\circ_1,\circ_2 \in \{<,>,\leq,\geq\}$ and $d_1,d_2 \in {\mathbb Q} \cup \{-\infty,+\infty\}$. 
Let us now discuss how to define the complement $V_1'$ of $V_1$; for the other sets the argument is symmetric. Note that $V_1'$ has the property that for all $(x',y') \in {\mathbb Q}^2$ and
$(x,y) \in V_1'$ if $x' \leq x$ and $y' \leq y$ 
then $(x',y') \in V_1'$. 

\input bijunctive.tex

\medskip 
\fbox{$(3)$ implies $(1)$.} 
Let $(u_1,v_1),(u_2,v_2),(u_3,v_3) \in R$
and let 
$$(u_0,v_0) := (\median(u_1,u_2,u_3),\median(v_1,v_2,v_3)) .$$ 
Let $\psi$ be a conjunct of a bijunctive definition
$\phi$ of $R$. 
We distinguish the following cases.  
\begin{itemize}
\item $\psi$ is a bound; in this case, preservation under the median operation is clear. 
\item $\psi$ is a linear inequality of the form $(ax + by) \circ c$. 
We first prove that $\psi$ is preserved when $\circ$ is $<$. 
Observe that, since 
$$\median(au_1,au_2,au_3) =
a \median(u_1,u_2,u_3)=au_0$$ 
there are at least two $i\in\{1,2,3\}$ such
that $au_0\le au_i$. Similarly there are at least two $i\in\{1,2,3\}$ such
that $bv_0\le bv_i$. By the pigeon hole principle we have at least one
value of~$i$ that satisfies both conditions simultaneously. For this
choice of~$i$
\[ au_0 + bv_0 \le au_i + bv_i < c \]
Hence, the median operation preserves all two-variable strict linear inequalities. 
It follows that the median operation also preserves 
all two-variable weak linear inequalities, 
since they are equivalent to infinite conjunctions
of two-variable strict linear inequalities:
$$\{(x,y \in {\mathbb Q}^2 \mid ax+by \leq c \} = \bigcap_{\epsilon>0} \{(x,y) \in {\mathbb Q}^2 \mid ax+by < c+\epsilon\}.$$
\item $\psi$ is a bend of the form $x \circ_1 d_1 \vee y \circ_2 d_2$. 
Then either 
\begin{itemize}
\item two out of $(u_1,v_1),(u_2,v_2),(u_3,v_3)$ satisfy $x \circ_1 d_1$, or 
\item two out of
$(u_1,v_1),(u_2,v_2),(u_3,v_3)$ satisfy $y \circ_2 d_2$. 
\end{itemize} 
In the first case we have that $(u_0,v_0)$
satisfies $x \circ_1 d_1$, in the second we
have that $(u_0,v_0)$
satisfies $y \circ_2 d_2$. 
\item $\psi$ is a bend of the form 
$x \circ_1 d_1 \vee (ax+by) \circ_2 d$. 
If two out of the three points $(u_i,v_i)$, for $i \in \{1,2,3\}$, satisfy $x \circ_1 d_1$, then $(u_0,v_0)$ also satisfies $x \circ_1 d_1$. 
So consider the case that $(u_1,v_1)$ and $(u_2,v_2)$ do not satisfy $x \circ_1 d_1$; then they must satisfy 
$(ax+by) \circ_2 d$. 
Note that $$(u_0,v_0) \in \{(u,v) \in {\mathbb Q}^2 \mid u_1 \leq u \leq u_2, v_1 \leq v \leq v_2\}$$
Suppose for contradiction that $(u_0,v_0)$ does not satisfy 
$(ax+by) \circ_2 d$. 
Then neither does
$(u_3,v_3)$; moreover, we claim that $(u_3,v_3)$ does not satisfy $x \circ_1 d_1$. 
Recall the assumption that $\circ_1 \in \{<,\leq\}$
if and only if $a > 0$. 
For concreteness, assume that $a>0$ and 
$b > 0$, and $u_1 \leq u_2$; the general case can be shown analogously. 
Set $x' := (d - b u_1) / a$ and note that 
$d_1 < x' \leq u_0$. Then we must have that 
$x' \leq u_3$ and hence $d_1 < u_3$. 
This is in contradiction to the assumption that 
$(u_3,v_3)$ satisfies $\psi$. 
\item $\psi$ is a bend of the form 
$(ax+by) \circ d \vee x \circ_2 d_2$; this case is similar to the previous one. 
\item $\psi$ is a bend of the form 
$(ax+by) \circ c \vee x \circ_1 d_1 \vee y \circ_2 d_2$ for $d_1,d_2 \in {\mathbb Q}$. 
If two points among $(u_1,v_1),(u_2,v_2),(u_3,v_3)$ satisfy the same clause in $\psi$,
then we can reduce to one of the previous two cases to deduce that $(u_0,v_0)$ satisfies $\psi$, too. Otherwise, we have that each of the points
$(u_1,v_1),(u_2,v_2),(u_3,v_3)$ satisfies precisely one of the three literals of $\psi$. Let $i_1,i_2,i_3 \in \{1,2,3\}$ be such that $(u_{i_1},v_{i_1})$ satisfies
 $(ax+by) \circ c$, $(u_{i_2},v_{i_2})$ satisfies
 $x \circ_1 d_1$, and $(u_{i_3},v_{i_3})$ satisfies
 $x \circ_2 d_2$. 
 Then the median 
 $(u_0,v_0)$ equals $(u_{i_1},v_{i_1})$ 
 and hence satisfies $\psi$. 
  
\end{itemize}
This concludes the proof of the implication $(3) \Rightarrow (1)$. 
\end{proof}

We remark that it was already known that a relation $R \subseteq D^k$ for a finite linearly ordered set $D$ is preserved by 
the median operation 
if and only if $R$ can be defined by a conjunction
of formulas of the form $x \circ_1 d_1 \vee y \circ_2 d_2$, for $\circ_1,\circ_2 \in \{\leq,\geq\}$ and $d_1,d_2 \in D$ (see~\cite{Disj,CCC}).
\red{Also note that Theorem~\ref{thm:syntax} implies
that existentially quantifying some variables in a bijunctive formula is equivalent to a (quantifier-free) bijunctive formula, because the existentially quantified formula is still preserved by the $\median$ operation and still semilinear.}

\section{Shostak's Theorem for Bends}
\label{sect:shostak}
Shostak's theorem 
characterises unsatisfiable TVPI constraints~\cite{Shostak}; we generalise it to a theorem characterising unsatisfiable conjunctions of bends.

\subsection{Composing bends}

\begin{lemma} \label{lem:compose-bends}
Let $\phi_1(x_0,x_1)$ and $\phi_2(x_1,x_2)$ be two bends where $x_1$ is a variable that is distinct from $x_0$ and $x_2$. Then 
$\psi := \exists x_1 (\phi_1(x_0,x_1) \wedge \phi_2(x_1,x_2))$ is equivalent to a bend $\phi(x_0,x_2)$. 
\end{lemma}
\begin{proof}
Suppose that for $i=1$ and $i=2$, the bend $\phi_i$ has the form 
$$x_{i-1} \circ_{i,1} d_{i,1} \vee (a_i x_{i-1} + b_i x_i \circ_i c_i) \vee x_i \circ_{i,2} d_{i,2}.$$
If one of $\phi_1$ or $\phi_2$ is equivalent to false, then so is $\psi$, so let us suppose that this is not the case. 
The formula 
$\phi_1 \wedge \phi_2$ is equivalent to
\begin{align}
& (x_0 \circ_{1,1} d_{1,1}  
\wedge x_1 \circ_{2,1} d_{2,1}) \label{eq:compose-bends} \\
\vee \; & (x_0 \circ_{1,1} d_{1,1}  \wedge
(a_2 x_1 + b_2 x_2 \circ_2 c_2))  \\
\vee \; & (x_0 \circ_{1,1} d_{1,1}  \wedge
x_{2} \circ_{2,2} d_{2,2})  \\
\vee \; & 
((a_1 x_{0} + b_1 x_1 \circ_1 c_1) \wedge x_1 \circ_{2,1} d_{2,1}) \\
\vee \; & 
((a_1 x_2 + b_1 x_1 \circ_1 c_1) \wedge (a_2 x_1 + b_2 x_2 \circ_2 c_{2}) ) \label{eq:main-disjunct} \\
\vee \; & 
((a_1 x_{0} + b_1 x_1 \circ_1 c_1)  \wedge x_{2} \circ_{2,2} d_{2,2}) \\
\vee \; & (x_1 \circ_{1,2} d_{1,2} \wedge 
x_{1} \circ_{2,1} d_{2,1}) 
\\
 \vee \; & (x_1 \circ_{1,2} d_{1,2} \wedge 
 (a_2 x_{1} + b_2 x_{2} \circ_{2} c_{2})) \\
\vee \; & (x_1 \circ_{1,2} d_{1,2}\wedge x_{2} \circ_{2,2} d_{2,2}) \label{eq:last-disjunct}
\end{align}

Assume that $b_1 < 0$; the case that $b_1>0$ is symmetric.
If $a_2 < 0$ then $\psi$ is equivalent to true because of disjunct~(\ref{eq:main-disjunct}).
Otherwise, $a_2 > 0$, and 
$\circ_{1,2} \in \{\geq,>\}$ and 
$\circ_{2,1} \in \{\leq,<\}$ by the definition of bends. 
Note that $a_2 x_1 + b_2 x_{2} \circ_{2} c_{2}$ is equivalent to $x_{1} \circ_{2} (c_{2} - b_2 x_{2})/a_2$. 
If additionally $a_1 x_{0} + b_2 x_{1} \circ_{1} c_{1}$, then $(c_1 - a_1 x_{0})/b_1 \circ_1 x_{1}$ since $b_1 < 0$. 
Hence, 
$$(c_1 - a_1 x_{0})/b_1 \circ (c_{2} - b_2 x_{2})/a_2$$ where $\circ \in \{<,\leq\}$ is strict if and only if one of $\circ_1,\circ_{2}$ is strict. 
 This is equivalent to 
 $$b_1 (c_{2} - b_2 x_{2}) \circ a_2 (c_1 - a_1 x_{0})$$
 which simplifies to
 $$a_1 a_2 x_{0} - b_1 b_2 x_{2} \circ a_2 c_1 -  b_1 c_{2} .$$ 
We use that existential quantification distributes over disjunction, and 
obtain that $\psi$ 
is equivalent to the bend
\begin{align}
x_{0} \circ_{1,1} d_{1,1} & \vee x_0 \circ_3 (c_1 - b_1 d_{2,1})/a_1 \label{eq:first} \\
& \vee (a_1 a_2 x_{0} - b_1 b_2 x_{2} \circ a_2 c_1 -  b_1 c_{2}) \\
x_{2} \circ_{2,2} d_{2,2} & \vee x_2 \circ_4 (c_2 - a_2 d_{1,2})/b_2 
\end{align}
where 
\begin{itemize}
\item $\circ_3 \in \{\leq,<\}$ if $a_1 > 0$ and $\circ_3 \in \{\geq,>\}$ if $a_1 < 0$; 
\item $\circ_3$ is strict if and only if at least one of $\circ_1,\circ_{2,1}$ is strict; 
\item $\circ_4 \in \{\leq,<\}$ if $b_2 > 0$ and $\circ_3 \in \{\geq,>\}$ if $b_2 < 0$; 
\item $\circ_4$ is strict if and only if at least one of $\circ_2,\circ_{1,2}$ is strict.  
\end{itemize}
\blue{We can now remove redundant bounds in the disjunction and obtain a bend.}
\end{proof}
The bend $\phi(x_0,x_2)$
in Lemma 4.1 is called the  \emph{residue bend} of $\exists x_1 (\phi_1(x_0,x_1) \wedge \phi_2(x_1,x_2))$. Note that if $\alpha(x_0) \vee \psi(x_0,x_2) \vee \beta(x_2)$ is the residue bend of $\psi(x_0,x_2)$ for distinct variables $x_0$ and $x_2$ then $\alpha$ 
is either the first or the second disjunct in (\ref{eq:first}). Then we define the \emph{source of $\alpha$ with respect to $(\phi_1,\phi_2)$} to be 
$x_0 \circ_{1,1} d_{1,1}$  in the first case  and $x_1 \circ_{2,1} d_{2,1}$ in the second case. 
The source of $\beta$ is defined analogously. 
Note that $\psi(x_0,x_2)$ only depends on the two-variable inequalities of $\phi_1$ and of $\phi_2$.

\subsection{Walks, Paths, and Cycles}
A \emph{walk} (from $x_0$ to $x_k$) is a sequence of bends 
$$W = \big (\phi_1(x_0,x_1),\phi_2(x_1,x_2),\dots,\phi_k(x_{k-1},x_k)\big).$$ We write $|W| := k$ for the length of $W$.
For a formula $\phi(x_1,x_2)$ we write
$\phi^{-1}(x_1,x_2)$ for the formula $\phi(x_2,x_1)$, and we write $W^{-1}$ for 
the walk $\big (\phi^{-1}_k(x_k,x_{k-1}),\dots,\phi^{-1}_k(x_1,x_0)\big)$. 
We make the convention that none of the
entries of $W$ is equivalent to a one-variable bend, unless $k=0$. 
If the variables $x_0,\dots,x_k$ are pairwise distinct, then $W$ is called a \emph{path}. 
If $W$ is a path 
then the \emph{residue} of $W$
is the formula $\res_W(x_0,x_k) := \exists x_1,\dots,x_{k-1} (\phi_1(x_0,x_1) \wedge \cdots \wedge \phi_k(x_{k-1},x_k))$ which is equivalent
to a bijunctive formula. 
A walk $W$ is \emph{closed} if $x_0 = x_k$.
We do allow the case that $k=0$, in which case the closed walk consists of a single bend $\psi(x_0,x_0)$. 
We call $W$ a \emph{cycle}
if the variables
$x_0,\dots,x_{k-1}$ are pairwise distinct.
If $W$ is a cycle then the \emph{residue} of $W$
is the formula $\res_W(x_0) := \exists x_1,\dots,x_{k-1} (\phi_1(x_0,x_1) \wedge \cdots \wedge \phi_k(x_{k-1},x_0))$.  
If $W$ is a cycle or a path, then inductively by Lemma~\ref{lem:compose-bends} the formula
\begin{align}
\exists x_1,\dots,x_{k-1} \bigwedge_{i \in \{1,\dots,k\}} \phi_i
\label{eq:residue}
\end{align}
is equivalent to a bend, which we call the
\emph{residue bend of $W$}, and denote
by $\res_W(x_0,x_k)$ in case where $W$ is a path, and by $\res_W(x_0)$ in case where $W$ is a cycle. 
Note that the residue bend of $W$ can be computed in polynomial time in the representation size of $W$ (all numbers appearing in the bends are represented in binary). 
The residue bend of a cycle $C$ starting and ending in $x$ is equivalent to $x \circ_1 d_1 \vee x \circ_2 d_2$ where $\circ_1 \in \{\leq,<\}$ and $\circ_2 \in \{\geq,>\}$, and $d_1,d_2 \in {\mathbb Q} \cup \{-\infty,+\infty\}$.

If the path or cycle $W$ consists of TVPI constraints, then the residue bend is either true, a bound, 
or again a TVPI constraint;
we then call it the \emph{residue inequality}. 


\begin{proposition}[Shostak~\cite{Shostak}]\label{thm:shostak}
A TVPI instance is unsatisfiable if and only if 
it contains a path $P$ 
from $x_0$ to $x_k$ and two cycles
with residue inequalities $\alpha(x_0)$ and
$\beta(x_1)$  
such that $\alpha(x_0) \wedge \res_P(x_0,x_k) \wedge \beta(x_k)$ is unsatisfiable. 
\end{proposition}

{\bf Note.} The terminology above is standard in graph theory today. Shostak's original formulation 
uses the terminology differently: his cycles are our closed walks, his simple cycles are our cycles,
his paths are our walks, and his simple paths are our paths. Moreover, his formulation 
says ``a TVPI instance is unsatisfiable iff its closure has an admissible simple cycle with unsatisfiable residue inequality'', where the closure is defined by adding all residue inequalities coming from (simple) cycles. The difference is due to the fact that in Shostak there is a pre-processing step, transforming all paths into cycles by the addition of one dummy variable. 
The admissibility requirements are made to avoid the situation where the residue inequality is a tautology (i.e., $x > -\infty$). We do not need these complications. We include  in Appendix~\ref{sect:proof-Shostak} a proof of Shostak's theorem for the convenience of the reader. 

For TVPI constraints, the notion of residue inequality can also be defined for walks. 
But note that if $W$ is just a walk of bends rather than a cycle or a path of bends, then the formula in~$(\ref{eq:residue})$ might no longer be equivalent to a single bend: consider
for example the closed walk $$(x \leq 0 \vee u \leq 1, u \geq 2 \vee x \geq 1, x \leq 2 \vee v \leq 1, v \geq 2 \vee x \geq 3)$$
which is equivalent to 
$(x \leq 0 \vee x \geq 1, x \leq 2 \vee x \geq 3)$. 
This formula 
is clearly not equivalent to a single bend. 
Note that a naive generalisation of Shostak's theorem fails for bends, as we see in the following example.

\begin{expl}\label{expl:worst-case}
Consider the bijunctive formula
$$x > 0 \wedge y < n \wedge \bigwedge_{i \in \{0,\dots,n\}} (x \geq i \vee y \leq i) \wedge x < z \wedge z < y$$
It is unsatisfiable since $$x < y \wedge x \geq 0 \vee y \leq 0 \wedge x \geq 1 \vee y \leq 1, \dots, x \geq n \vee y \leq n$$ 
implies $x \geq 0 \vee y \leq n$, so $\Phi$ is unsatisfiable because of its conjuncts $x > 0$ and $y < n$.
But removing any bend in the formula above 
results in a satisfiable formula, so there cannot
be simple cycles and a simple path as in the statement of Shostak's theorem witnessing unsatisfiability of $\Phi$. 
\end{expl} 

\subsection{Handcuffs}
For our generalisation of Shostak's theorem for bends the following terminology is convenient. 

\begin{definition}[handcuff]
Let $C$ be a cycle starting and ending in $x$, 
let $P$ be a path from $x$ to $y$,
and let $D$ be a cycle starting and ending in $y$,
such that $C$ and $D$ have disjoint sets of variables, $x$ is the only variable shared by $C$ and $P$, and $y$ is the only variable shared by $P$ and $D$. 
Then the walk $(C,P,D)$ is called a \emph{handcuff}. 
We say that the handcuff is \emph{unsatisfiable} if the conjunction over all bends in the handcuff is unsatisfiable. 
\end{definition}
\begin{figure}[h]
  \begin{center}
  \includegraphics[scale=0.5]{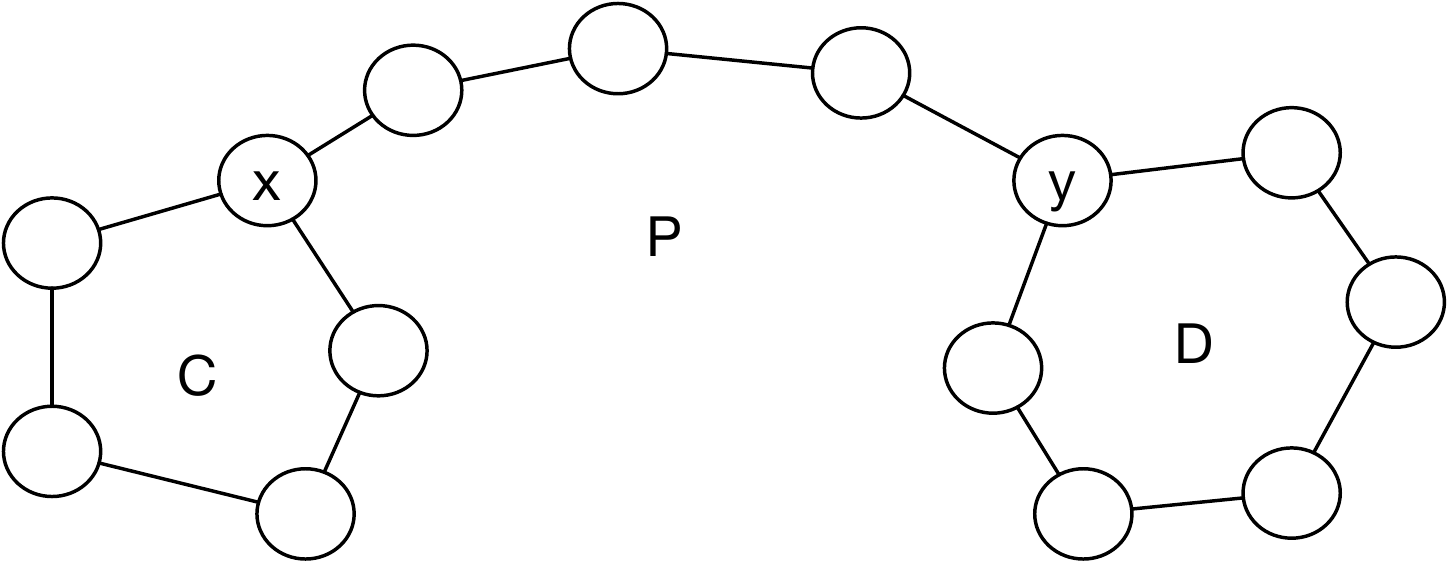}
  \end{center}
\caption{An illustration of a handcuff $(C,P,D)$.}
  \label{fig:handcuff}
\end{figure}
Note that a handcuff $(C,P,D)$ is unsatisfiable
if and only 
if $\res_C(x) \wedge \res_P(x,y) \wedge \res_D(y)$
is unsatisfiable. 
If $\Phi$ is a bijunctive formula with variables $V$  and $W$ is a walk with variables $U$ then a function 
$h \colon U \to V$ is called a \emph{homomorphism} if  for every bend $\phi(x,y)$ that appears in $W$ the bend $\phi(r(x),r(y))$
is a conjunct of $\Phi$.

\begin{definition}[handcuff refutation]
Let $\Phi$ be a bijunctive formula with $n$ variables.  A \emph{handcuff refutation of $\Phi$} 
is a homomorphism from an unsatisfiable handcuff $(C,P,D)$ to $\Phi$. 
\end{definition}

Clearly, if a bijunctive formula has a handcuff refutation, it is unsatisfiable. 
The converse is false, as demonstrated in Example~\ref{expl:worst-case}. 
\red{Also note that Shostak's theorem implies that a TVPI instance $\Phi$ is unsatisfiable if and only if it has a handcuff refutation.
When $\phi$ is a conjunct of $\Phi$
and $\beta$ is another formula, we write
$\Phi[\phi/\beta]$ for the formula obtained
from $\Phi$ by replacing $\phi$ by $\beta$.}

\begin{definition}[handcuff consistency]
Let $\Phi$ be a bijunctive formula with variables $V$. 
We say that $\Phi$ is \emph{handcuff consistent} 
if \red{for every bound $\beta$ 
from a bend $\phi \in \Phi$ the formula
$\Phi[\phi/\beta]$ does not have a handcuff refutation. }
\end{definition}

\begin{theorem}[Shostak's theorem for bends]
\label{thm:bend-shostak}
Let $\Phi$ be a bijunctive formula which is handcuff consistent. Then $\Phi$ is satisfiable. 
\end{theorem}

\begin{proof}
Let $V = \{x_1,x_2,\dots,x_n\}$ be the variables of $\Phi$. 
\red{Suppose first that $\Phi$ contains a bend $\phi$ which
is not a TVPI constraint; so the bend must have
a literal which is a bound $\beta(x)$. 
Since $\Phi$ is handcuff consistent, the  instance $\Psi := \Phi[\phi/\beta]$ does not have a handcuff refutation. 
However, $\Psi$ might not be handcuff consistent, i.e., $\Psi$ might have a bend $\psi(x,y)$ with a bound $\chi(x)$ as literal such that $\Psi[\psi/\chi]$ has a handcuff refutation;
we can assume that
this refutation is a homomorphism from 
an unsatisfiable handcuff of the form $(E,Q,\chi)$. 
Let $\psi'$ be obtained from $\psi$ by
removing $\chi$, and let $\Psi' := \Psi[\psi/\psi']$.} 

\red{{\bf Claim.} $\Psi'$ does not have a handcuff refutation. 
Suppose that otherwise there is an unsatisfiable handcuff $(C,P,D)$
with a homomorphism to $\Psi$. 
Since $\Psi$ does not have a handcuff refutation,
we may assume that $\psi'(x,y)$ appears in either $P$ or in $C$. 
In the first case, $P$ can be written 
as $(P_1,\psi',P_2)$. 
But then $(E,(Q,\psi,P_2),D)$ is an unsatisfiable
handcuff and homomorphically maps to 
$\Psi$, contrary to the assumption that $\Psi$
does not contain a handcuff refutation. See
Figure~\ref{fig:shostak} for an illustration. 
In the second case, $C$ can be written 
as $(C_1,\psi',C_2)$. 
But then $(E,(Q,\psi,C_2,P),D)$ is an unsatisfiable
handcuff and homomorphically maps to 
$\Psi$, contrary to the assumption that $\Psi$
does not contain a handcuff refutation.}


\red{Again, $\Psi'$ might not be handcuff consistent, and we continue to remove literals from bends. In this way, we eventually end up with a handcuff-consistent bijunctive formula which is equivalent to $\Psi$. If this formula contains a bend with more than one literal,
we continue as above, until
eventually we obtain a TVPI instance.
The statement then follows from Shostak's theorem (Theorem~\ref{thm:shostak}; actually, Lemma~\ref{lem:shostak} suffices).}
\end{proof}

\begin{figure}[h]
  \begin{center}
  \includegraphics[scale=0.6]{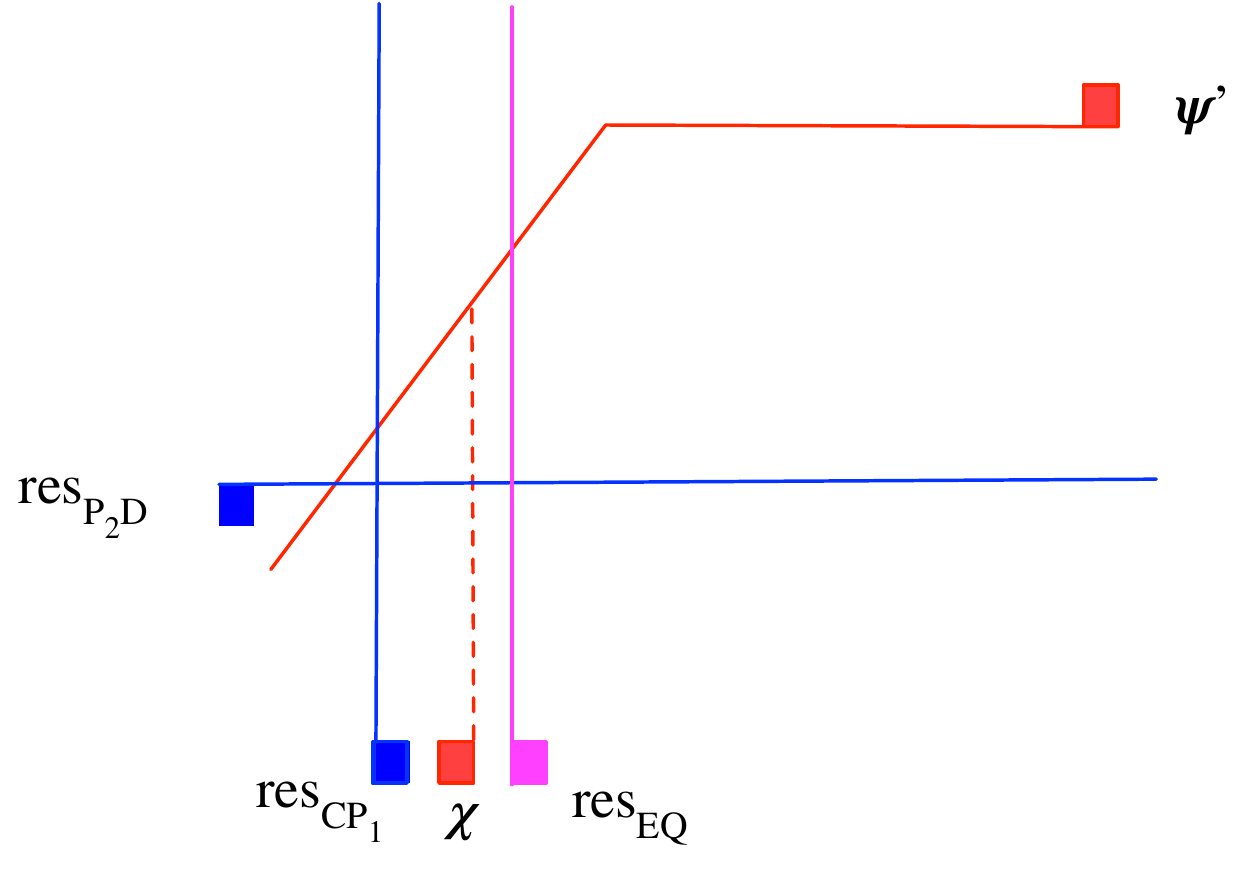}
  \end{center}
\caption{Illustrations for the proof of Theorem~\ref{thm:bend-shostak}.}
  \label{fig:shostak}
\end{figure}

\subsection{Small Handcuff Refutations}
\red{Note that there is no length restriction on
the size of the unsatisfiable handcuff in a handcuff refutation. To apply Theorem~\ref{thm:bend-shostak} to prove the correctness of our algorithm, we need such a bound. For this purpose, we first prove 
a strengthening of the contraposition of this theorem in an important special case (Lemma~\ref{lem:cycle}).} 
Let $\alpha$ be a bound in the residue bend of a path $P$. 
 The \emph{source of $\alpha$ with respect to $P$} is defined inductively:
 \begin{itemize}
\item if $P$ consists of just one bend,
then the residue inequality equals this bend, and the source of $\alpha$ is defined to be $\alpha$; 
\item 
if $P$ consists of just two bends we have already defined the source of $P$;
\item otherwise, suppose that 
$P = (Q,\phi(x_{k-1},x_k))$ 
is a path from $x_0$ to $x_k$ 
and let $\psi(x_0,x_{k-1})$ be the residue bend of $Q$.
Then the residue bend of $P$
equals the residue bend of $\exists x_{k-1} (\psi(x_0,x_{k-1}) \wedge \phi(x_{k-1},x_k))$. If the source of $\alpha$ with respect to $(\psi,\phi)$ lies in $\psi$, then we define inductively the source of $\alpha$ to be the source of $\psi$ with respect to $Q$. Otherwise, the source $\beta$ of $\alpha$ with respect to $(\psi,\phi)$ lies in $\phi$, 
and we define the source of $\alpha$ with respect to $P$ to be $\beta$, too. 
\end{itemize}


\begin{lemma}\label{lem:cycle}
Let $C = (\phi_0(x_0,x_1),\phi_1(x_1,x_2),\dots,\phi_n(x_n,x_0))$ be a cycle of bends
and let $\Phi := \phi_0(x_0,x_1) \wedge \dots \wedge \phi_n(x_n,x_0)$. 
Let $\alpha$ and $\beta$ be a lower
and upper bound such that 
$\alpha(x_0) \wedge \Phi \wedge \beta(x_0)$
is unsatisfiable. Then 
$\Phi \wedge \beta(x_0)$ or $\Phi \wedge \beta(x_0)$ is not handcuff consistent.
\end{lemma}
\begin{proof}The statement is trivial
if $\res_C(x_0)$ is equivalent to a bound. Otherwise, 
the residue inequality $\res_C(x_0)$ must be of the form $\gamma \vee \delta$ for
bounds $\gamma$ and $\delta$
such that $\gamma \wedge \alpha$ is unsatisfiable
and $\delta \wedge \beta$ is unsatisfiable. 
\red{Suppose the latter case applies; the other case is similar. 
Let $P$ be the 
path $(\phi_0(x_0,x_1),\dots,\phi_n(x_n,x_{n+1}))$ where $x_{n+1}$ is a new variable. 
Let $\eta$ be the source of 
$\delta$ with respect to $P$, and suppose $\eta$ is a bound of $\phi_i$. 
Then $(\eta,(\phi_{i+1},\dots,\phi_n),\beta)$ or $(\beta,(\phi_{1},\dots,\phi_{i-1}),\eta)$
is an unsatisfiable handcuff in $\Phi \wedge \eta$, which implies the statement.}
\end{proof}

\begin{corollary}\label{cor:short-handcuff}
\red{A bijunctive formula $\Phi$ is handcuff consistent 
if and only if for every bound $\beta$ from a bend $\phi \in \Phi$ the formula $\Phi[\phi/\beta]$ does not have a homomorphism from an unsatisfiable handcuff $(C,P,D)$ to $\Phi$ such that $|C|,|P|,|D| \leq n$.}
\end{corollary}
\begin{proof}
\red{Clearly, handcuff consistency implies the given condition in the statement. Now suppose that 
$\Phi$ is not handcuff consistent, i.e., 
there is a literal $\beta$ in a bend $\phi$ of $\Phi$
such that there is a homomorphism $h$ from an unsatisfiable handcuff $(C,P,D)$ to $\Psi := \Phi[\phi/\beta]$.
Pick $\beta$ such that $|CPD|$ is shortest possible. If $|C|,|P|,|D| \leq n$ there is nothing to show.
Suppose first that $|P| > n$; then there must be variables $z_1,z_2$ that appear in $P$ such that $h(z_1)=h(z_2)$. 
Write $P$ as $P_1QP_2$ where
$Q$ is a path from $z_1$ to $z_2$. 
Let $E$ be the cycle obtained from $Q$ by replacing both $z_1$ and $z_2$ by a new variable $z$. Note that
$\res_{CP_1}(z) \wedge \res_{E}(z) \wedge \res_{P_2D}(z)$ is unsatisfiable.  Lemma~\ref{lem:cycle} implies that the conjuncts of $CP_1E$ or the conjuncts of $EP_2D$ are not handcuff consistent. In both cases we obtain a contradiction to the choice of $\beta$ so that $|CPD|$ is shortest possible. 
Now suppose $|C| > n$; the case that $|D| > n$ is analogous. Then there must be variables $z_1,z_2$ that appear in $C$ such that $h(z_1)=h(z_2)$. 
Write $C$ as $C_1QC_2$ where
$Q$ is a path from $z_1$ to $z_2$. 
Let $E$ be the cycle obtained from $Q$ by replacing both $z_1$ and $z_2$ by a new variable $z$. Note that
$\res_{DPC_1}(z) \wedge \res_{E}(z) \wedge \res_{DPC_2}(z)$ is unsatisfiable.  Lemma~\ref{lem:cycle} implies that the conjuncts of $DPC_1E$ or the conjuncts of $DPC_2E$ are not handcuff consistent. In both cases we obtain a contradiction to the choice of $\beta$ so that $|CPD|$ is shortest possible.} 
\end{proof}

\section{An algorithm for median-closed constraints}
\label{sect:alg}
This section presents a (strongly) polynomial  algorithm for deciding whether a given bijunctive formula $\Phi$ is satisfiable over the rational numbers. 
The overall structure of our algorithm is similar to the structure of the algorithm of Hochbaum and Naor~\cite{HochbaumNaor} for TVPI constraints. 
A key subprocedure of their algorithm is
a procedure that Hochbaum and Naor credit to Aspvall and Shiloach~\cite{AspvallShiloach}. 
We mention that the procedure is not presented in~\cite{AspvallShiloach}, 
but can be found implicitly in the PhD thesis of Aspvall~\cite{Aspvall}, 
as has been noted already in~\cite{RestrepoWilliamson}. 

The required subprocedure 
tests for a given set of TVPI constraints $\Psi$,
a value $s \in {\mathbb Q}$, and a variable $x$ of $\Psi$ whether
$\Psi \wedge x \geq s$ is satisfiable; if $\Psi$ is unsatisfiable, the procedure can answer arbitrarily. 
One of the contributions of this work is a  generalisation of this procedure from TVPI constraints to bijunctive formulas $\Psi$. 
The new procedure is called PROPAGATE and is described in Section~\ref{sect:as-gen}. We then 
explain how to use the procedure PROPAGATE to decide satisfiability of
$\Psi$ in Section~\ref{sect:hn-gen}.

To describe the procedures in more detail, 
first observe that
deciding the following tasks can be done in
polynomial time (the first three even in constant time), assuming unit cost for performing the arithmetic operations addition, multiplication, and size comparison: 
\begin{enumerate}
\item deciding whether a bound implies another bound;
\item computing the strongest bound on a given variable which is implied by the 
conjunction of a bend and two bounds;
\item deciding whether the conjunction of two bounds is unsatisfiable;
\item \blue{computing the residue bend of a cycle of bends.}
\end{enumerate}
We will therefore use these tasks freely in the pseudocode of our algorithms.

\subsection{Generalising Aspvall-Shiloach}
\label{sect:as-gen} 
In this section we describe a procedure that
tests whether $\Phi \wedge x \geq s$ is satisfiable 
for a given satisfiable bijunctive formula $\Phi$ with variables $V$. 
 The idea of the algorithm is to propagate bounds
 along constraints to obtain stronger and stronger bounds on the variables that are implied by $\Phi \wedge x \geq s$; if we find a contradiction in this way, then $\Phi \wedge x \geq s$ is clearly unsatisfiable. 
 However, the procedure might not terminate while deriving stronger and stronger bounds; 
 a simple unsatisfiable example of this type
  is $y \geq 2x \wedge x \geq 2y \wedge x \geq 1$
  where we can derive the bounds $x \geq 1$, 
  $y \geq 2$, $x \geq 4$, $y \geq 8$, $x \geq 16$ etc. 
   
To get around this problem, the algorithm
also uses the bound propagation to detect cycles  $(\phi_0(u_0,u_1),\dots,\phi_k(u_{k-1},u_0))$ 
of bends in $\Phi$, 
and then uses these cycles to symbolically
compute better bounds on $u_0$, if possible. 
In the example above,
we would use the sequence $(2x \leq y, 2y \leq x)$
to deduce the bound $x \leq 0$. 
Aspvall showed for TVPI constraints $\Phi$ that if we do not find a contradiction after propagating bounds for $3n$ steps in this way, where $n$ is the number of variables in $\Phi$, then $\Phi \wedge x \geq s$ is satisfiable (essentially Lemma 9 in~\cite{Aspvall}). The proof is based on Shostak's theorem.

This step is for bends rather more complicated than the corresponding step for TVPI constraints in Aspvall's algorithm. 
However, we can still use bound propagation 
to detect cycles of bends, and to establish 
handcuff consistency in this way;
we can then use our generalisation of Shostak's
theorem to prove the satisfiability 
of $\Phi \wedge x \geq s$. 
In the following, when we write \emph{stronger} we mean strictly stronger. 

\begin{lemma}\label{lem:bend-discover}
Let $P$ be a path of bends from $x$ to $y$ and $\beta(x)$ a bound on $x$. If $\beta(x) \wedge \res_P(x,x)$ is unsatisfiable then 
\begin{enumerate}
\item $\res_P(x,x)$ is unsatisfiable, or
\item $\beta(x) \wedge \res_P(x,y)$ implies a stronger bound on $y$
than $\beta(y)$, or
\item $\beta(y) \wedge \res_P(x,y)$ implies a stronger bound on $x$
than $\beta(x)$, or
\item $\beta(x) \wedge \res_P(x,y) \wedge \beta(y)$ is unsatisfiable. 
\end{enumerate}
\end{lemma}

\begin{figure}[h]
  \begin{center}
  \includegraphics[scale=0.4]{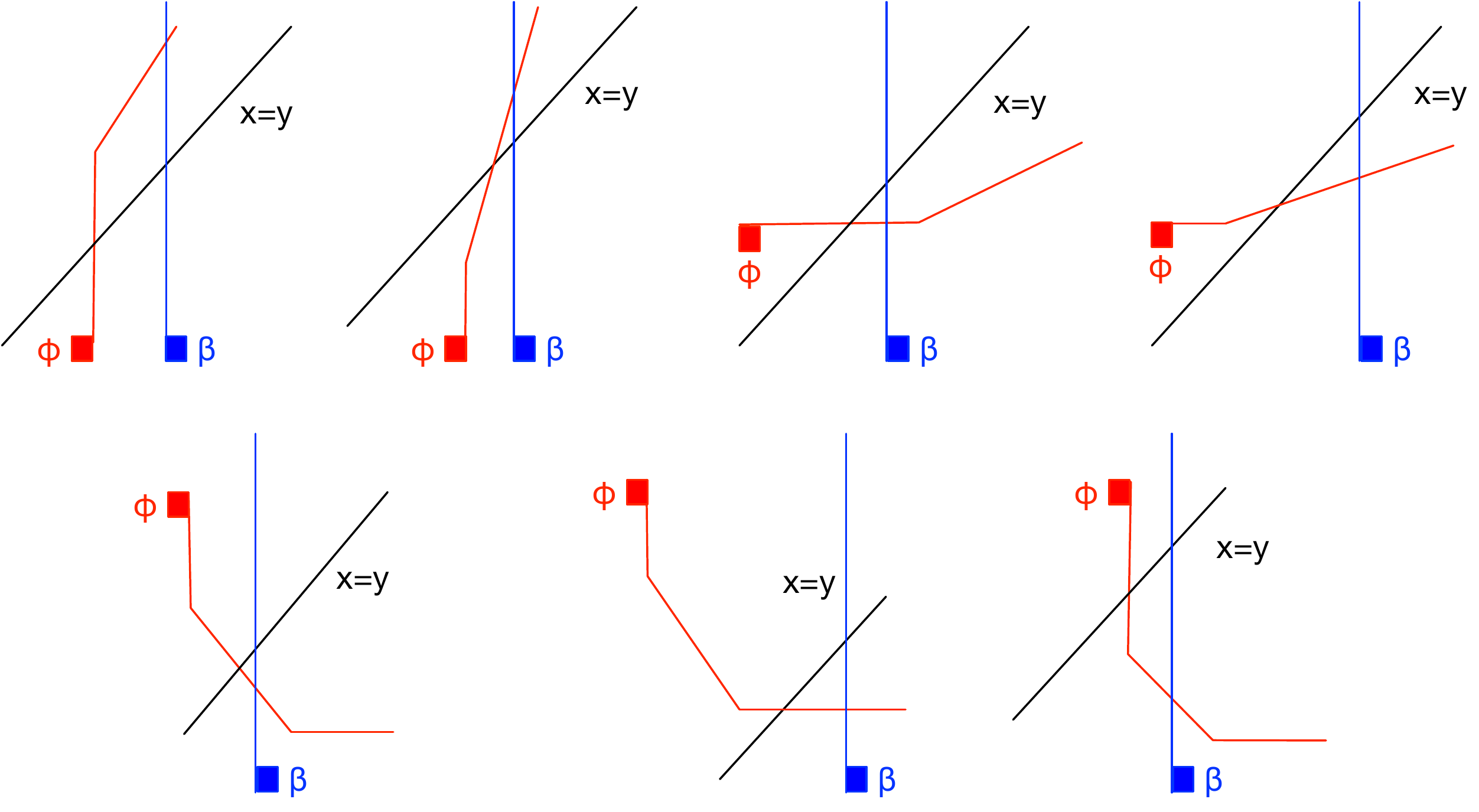}
  \end{center}
\caption{Illustrations of $\phi := \res_P(x,y)$ and $\beta(x)$ in the proof of Lemma~\ref{lem:bend-discover}.}
  \label{fig:aspvall-1}
\end{figure}

\begin{proof}
\blue{Assume that $\beta(x)$ is a lower bound $x \circ g$ for $g \in {\mathbb Q}$ and 
$\circ \in \{\geq,>\}$;} the case that $\beta(x)$ is an upper bound is analogous. 
Suppose that $\res_P(x,x)$ is satisfiable. 
Since $\beta(x) \wedge \res_P(x,x)$ is unsatisfiable, $\res_P(x,x)$ 
must be equivalent to an upper bound $\gamma(x)$. 
Then the residue bend 
$\res_P(x,y)$ is of one of the following forms (see Figure~\ref{fig:aspvall-1}):

\begin{enumerate}
\item $x \circ_1 d \vee ax+by \circ_2 c$ 
for $a,b \in {\mathbb Q}$,  $\circ_1, \circ_2 \in \{\leq,<\}$, 
$c,d \in {\mathbb Q} \cup \{-\infty\}$, 
$a > 0$, $b < 0$, 
\blue{$-\frac{a}{b} \geq 1$}, 
(the two upper left pictures in Figure~\ref{fig:aspvall-1});
\item $ax+by \circ_2 c \vee y \circ_3 d$, 
for $a,b \in {\mathbb Q}$, $\circ_2, \circ_3 \in \{\leq,<\}$, 
$c,d \in {\mathbb Q} \cup \{-\infty\}$, $a < 0$, 
$b > 0$, 
\blue{$-\frac{a}{b} \leq 1$}, 
(the two upper right pictures in Figure~\ref{fig:aspvall-1});
\item $x \circ_1 d_1 \vee ax+by \circ_2 c \vee y \circ_3 d_2$ where $\circ_1,\circ_2,\circ_3 \in \{\leq,<\}$, $a,b \in {\mathbb Q}$, $c,d_1,d_2 \in {\mathbb Q} \cup \{-\infty\}$, $a,b > 0$, and 
(the lower three pictures in Figure~\ref{fig:aspvall-1}). 
\end{enumerate}

Note that in all three cases, $a+b \geq 0$. Since $x \circ g \wedge \res_P(x,x)$ is unsatisfiable, 
\blue{$(a+b)g \circ' c$ where $\circ' \in \{>,\geq\}$ is strict if and only if
both $\circ$ and $\circ_2$ are non-strict. In the following, we assume that all inequalities in the constraints are non-strict; the adaptation to the general case is straightforward but notationally cumbersome.} 

\red{In the first case, $x \geq g \wedge \res_P(x,y)$ implies 
$g \leq x \leq \frac{c-by}{a}$ and hence 
$y \geq \frac{c-ag}{b}$. This is stronger} than $y \geq g$
since $\frac{c-ag}{b} > g$ if and only if 
$c < gb+ga$ which is true as we have noted above. 
So item $(2)$ applies. 
 
\red{In the second case, $y \geq g \wedge \res_P(x,y)$ implies 
$g \leq y \leq \frac{c-ax}{b}$
and hence $x \geq \frac{c-bg}{a}$. 
This is stronger 
than $x \geq g$
since $\frac{c-bg}{a} > g$ if and only if 
$c < bg+ag$, 
which is true as we have noted above. So item (3) applies.}

In the third case, $x \geq g \wedge \res_P(x,y)$ implies 
$y \leq \frac{c-ag}{b}$ 
and we obtain that $x \geq g \wedge \res_P(x,y) \wedge y \geq g$ is unsatisfiable: 
$g \leq \frac{c-ag}{b}$ holds 
if and only if $ag+bg \leq c$ which is false as we have noted above. 
So item $(4)$ applies. 
\end{proof}

The following lemma is illustrated in  Figure~\ref{fig:aspvall-2}. 

\begin{lemma}\label{lem:bend-detect}
Let $C = (\phi_1(x,u_1),\dots,\phi_k(u_{k-1},y))$ and $D = (\psi_1(x,v_1),\dots,\psi_l(v_{l-1},y))$ be paths between distinct variables $x,y$ and let $\beta(x)$ be a bound. 
Suppose that $\beta(x) \wedge \phi_C(x,x)$ is unsatisfiable and that the strongest bound on $y$ implied by $\beta(x) \wedge \phi_D(x,y)$ is stronger 
than the strongest bound on $y$ implied by
$\beta(x) \wedge \phi_C(x,y)$. 
Then 
\begin{enumerate}
\item $\psi_1(x,v_1) \wedge \cdots \wedge \psi_l(v_{l-1},x) \wedge \beta(x)$ is 
not handcuff consistent, or
\item $\phi_1(x,u_1) \wedge \cdots \wedge \phi_k(u_{k-1},x)  \wedge \psi_1(x,v_1) \wedge \cdots \wedge \psi_l(v_{l-1},x)$ is not handcuff consistent.
\end{enumerate}
\end{lemma}

\begin{figure}[h]
  \begin{center}
  \includegraphics[scale=0.4]{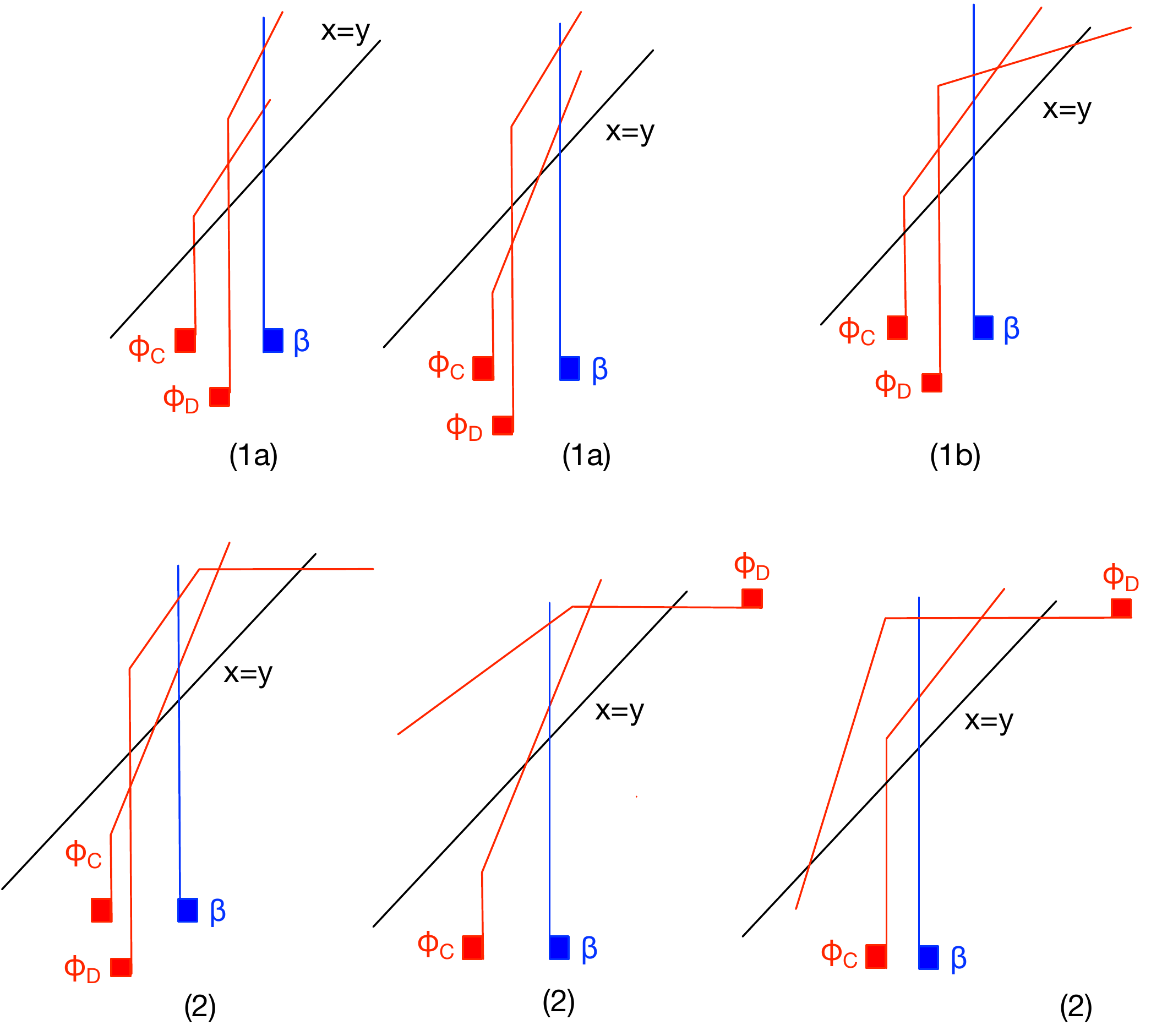}
  \end{center}
\caption{Illustrations of $\phi_C(x,y)$, $\phi_D(x,y)$, and $\beta(x)$ in the proof of Lemma~\ref{lem:bend-detect}.}
  \label{fig:aspvall-2}
\end{figure}

\begin{proof}
Assume that $\beta(x)$ is a lower bound; the case that $\beta(x)$ is an upper bound is analogous.

Let $x \circ_1 d_1 \vee ax + by \leq c \vee y \circ_2 d_2$ be the residue bend $\phi_C(x,y)$, for $\circ_1,\circ_2 \in \{\leq,<,\geq,>\}$ and $a,b \in {\mathbb Q}$, $c \in {\mathbb Q} \cup \{-\infty\}$. 
Since $\beta(x) \wedge \res_P(x,x)$ is unsatisfiable, we can assume that $a+b \geq 0$, as in the proof of the previous lemma. 
Let $x \circ_1' d'_1 \vee a'x + b'y \circ' c' \vee y \circ'_2 d'_2$ be the residue bend $\phi_D(x,y)$, for $\circ'_1,\circ'_2,\circ' \in \{ \leq,<,\geq,>\}$, $a',b' \in {\mathbb Q}$ and  $c' \in {\mathbb Q} \cup \{-\infty\}$. 

If $a' < 0$ then $a'x+b'y \circ' c'$ 
is equivalent to $\frac{c' - b' y}{a'} \circ' x$ and $\circ'_1 \in \{\geq,>\}$ by the definition of bends. 
Since both $\beta(x)$ and $x \circ'_1 d_1$ are lower bounds, the strongest bound on $y$ implied by
$\beta(x) \wedge x \circ'_1 d'_1 \vee a'x+b'y \circ' c' \vee y \circ'_2 d'_2$ 
is trivial too, contrary to our assumptions.  
So let us assume that
$a' > 0$ and therefore, by the definition of bends,
that $\circ_1 \in \{\leq,<\}$. 
Again, since the strongest lower
bound on $y$ implied by $x \geq 0 \wedge x \circ_1' d_1' \vee a'x+b'y \circ' c' \vee y \circ_2' d_2'$ must be non-trivial, we get $b' < 0$ and $\circ_2 \in \{\geq,>\}$. 

We claim that $b < 0$. Otherwise, 
 if $b > 0$ then $\circ_2 \in \{\leq,<\}$
 and 
  $\beta(x) \wedge \phi_C(x,y)$  would imply an upper bound on $y$
in contradiction to the assumption that 
the strongest bound on $y$ implied by 
$\beta(x) \wedge \phi_D(x,y)$ is a lower bound, and supposed to be stronger. We distinguish the following cases. 
\begin{enumerate}
\item[(1)] \blue{$d'_2 = +\infty$ 
so that $y \circ_2' d'_2$ is equivalent to false.
In this case 
$\phi_D(x,x) \wedge \beta(x)$ is unsatisfiable (see the illustrations in Figure~\ref{fig:aspvall-2} labelled with (1a)) or otherwise not handcuff consistent (labelled with (1b) in Figure~\ref{fig:aspvall-2}))}. 
\item[(2)] \blue{$d'_2 < +\infty$ (see the illustrations labelled with (2) in Figure~\ref{fig:aspvall-2}). 
We claim that in this case
 $\phi_1(x,u_1) \wedge \cdots \wedge \phi_k(u_{k-1},x)  \wedge \psi_1(x,v_1) \wedge \cdots \wedge \psi_k(v_{l-1},x)$
is not handcuff consistent, so we have established item $(2)$ of the statement. }
\end{enumerate}
This concludes the proof. 
\end{proof}

The procedure PROPAGATE is given in Figure~\ref{fig:alg-as-gen}. It uses the following terminology: 
A literal $\psi$ in a bend $\phi(u,v)$ from $\Phi$
is called \emph{redundant} if $\phi \wedge \beta_u^{\low} \wedge \beta_u^{\high} \wedge \beta_v^{\low} \wedge \beta_v^{\high}$
is equivalent to $\phi' \wedge \beta_u^{\low} \wedge \beta_u^{\high} \wedge \beta_v^{\low} \wedge \beta_v^{\high}$ where $\phi'$ is obtained from $\phi$ by removing $\psi$. 
The main loop of the algorithm is executed until the number of redundant literals of $\Phi$ does not change; clearly, this condition can be checked efficiently. Within the main loop, 
the algorithm performs bound propagations for at most $2 \cdot |V|$ rounds, where $V$ is the set of variables of $\Phi$; the motivation for this is that 
the algorithm needs to find paths from $x$ to some variable $y$ and a cycle starting and ending in $y$, and the total number of vertices 
on the path and the cycle is bounded by $2 \cdot |V|$. \red{If we detect a closed walk $D$ starting and ending in $x$ after having performed the bound propagation, we transform $D$ into a cycle $D'$ by replacing all but the first and the last variable on the walk by fresh variables. 
This residue bend of $D'$ will be used 
to further improve the bounds $\beta_v^{\high}$ and $\beta_v^{\low}$ (lines 20-23).} 

\begin{figure}[h]
\fbox{
\begin{tabular}{l}
// Input: a bijunctive formula $\Phi$ with variables $V$, $x \in V$, and $s \in {\mathbb Q}$. \\
// Task: If $\Phi$ is satisfiable, decide whether
$\Phi \wedge x \geq s$ has a solution. \\
// If $\Phi$ is unsatisfiable, any answer is fine. \\
\\
01: For all $u \in V$ do \\
02: \hspace{.3cm} Let $\beta_u^{\low}$ be the most restrictive lower bound on $u$ in $\Phi$. \\
03: \hspace{.3cm} Let $\beta_u^{\high}$ be the most restrictive upper bound on $u$ in $\Phi$. \\
04: \hspace{.3cm} For all $i \in \{1,\dots,2 \cdot |V|\}$, let $P^{\low}_{u,i}$ and $P^{\high}_{u,i}$ be undefined. \\
05: Replace $x^{\low}$ by `$x \geq s$'
if this is more restrictive. \\
06: Do \\
07: \hspace{.3cm} 
For $i = 1,\dots,2 \cdot |V|$ do \\
08: \hspace{.6cm} 
For each bend 
$\phi(u,v) \in \Phi$:  \\
09: \hspace{.9cm} Let $\beta$ be the strongest bound on $v$ implied by $\beta_u^{\low} \wedge \beta_u^{\high} \wedge \phi(u,v)$. \\
10: \hspace{0.9cm} If $\beta$ is more restrictive than $\beta^{\low}_v$: \\
11: \hspace{1.2cm} Replace $(\beta_v^{\low},P^{\low}_{v,i})$  by $(\beta,\phi(u,v))$; \\ 
13: \hspace{1.2cm} If $\beta_v^{\low} \wedge \beta_v^{\high}$ is unsatisfiable 
then answer `No'; \\
14: \hspace{.9cm} Analogous steps to 10-13 for $\beta_v^{\high},P_{v,i}^{\high}$ instead of $\beta_v^{\low},P_{v,i}^{\low}$.  \\
\blue{15: \hspace{.3cm} For all $v \in V$ do} \\
\blue{16: \hspace{.6cm} Let $u := v$; $D := ()$; $i := 2 \cdot |V|+1$} \\
17: \hspace{.6cm} While there exists $j < i$ such that $P_{u,j}^{\low}$ is defined: \\
18: \hspace{.9cm} Let $j < i$ be largest such that $\phi(w,u) := P_{u,j}^{\low}$ is defined.\\
19: \hspace{.9cm} Let $u := w$, $D := (\phi,D)$, $i:=j$; \\
20: \hspace{.9cm} If $w=v$ then \\
21: \hspace{1.2cm} let $D'$ be the cycle starting at $v$ obtained from $D$ by renaming variables; \\
22: \hspace{1.2cm} replace $\beta^{\low}_{v}$ by the strongest lower bound implied by $\phi_{D'}(v) \wedge \beta^{\low}_{v}$; \\
23: \hspace{1.2cm} replace $\beta^{\high}_{v}$ by the strongest upper bound implied by $\phi_{D'}(v) \wedge \beta^{\high}_{v}$. \\
23: \hspace{.6cm} Analogous steps to 16-23 for
$P_{v,i}^{\high}$ instead of $P_{v,i}^{\low}$.  \\
24: Loop until number of redundant literals in $\Phi$ does not change. \\
25: Answer `Yes'.
\end{tabular}
}
\caption{The procedure PROPAGATE.} 
\label{fig:alg-as-gen}
\end{figure}

\medskip 
{\bf Correctness of PROPAGATE.} 
Suppose that $\Phi$ is satisfiable, since otherwise there is nothing to be shown. 
Note that at each time of the execution of the algorithm and for every $u \in V$ we have that $\Phi \wedge x \geq s$ implies 
$\beta_u^{\low}$ and $\beta_u^{\high}$; this can be shown by a straightforward induction.
Hence, if the algorithm finds that $\beta_v^{\high} \wedge \beta_v^{\low}$ is unsatisfiable, then the answer `No' is
correct. 

Now suppose that the algorithm answers
`Yes'. We have to show that in 
this case $\Phi \wedge x \geq s$ is satisfiable. 
By Theorem~\ref{thm:bend-shostak}, it suffices
to show that at the final stage of the algorithm, 
the instance $\Psi$ obtained from $\Phi \wedge x \geq s$ by removing all redundant literals is handcuff consistent; \red{we use the characterisation of handcuff consistency from Corollary~\ref{cor:short-handcuff}}. 
We first show that $\Psi$ itself does not have a handcuff refutation. 
Otherwise, 
since $\Phi$ is satisfiable, 
$\Psi$ has a handcuff refutation 
that involves the conjunct $x \geq s$. By the convention that paths and cycles with at least two variables do not contain one-variable bends 
(and by Corollary~\ref{cor:short-handcuff}) 
we can assume without loss of generality that 
\red{the unsatisfiable handcuff that has a homomorphism $r$ to $\Psi$ is of the form
$((x \geq s),P,C)$ and $|P|,|C| \leq |V|$.}
Suppose that $P$ is a path from $x$ to $y$,
and let $\beta(y)$ be the strongest bound  implied by $x \geq s \wedge \res_P(x,y)$. 
We assume that $\beta$ is a lower bound; the argument when $\beta$ is an upper bound is analogous. Note that then $\phi_C(y)$ must be an upper bound. 
After at most $|V|$ iterations of the inner loop 
the algorithm will update \red{$\beta^{\low}_{r(y)}$} with
a bound that is at least as strong as the bound $\beta$ above. After at most $|V|$ more iterations the algorithm updates \red{$\beta^{\low}_{r(y)}$} again; this follows from
Lemma~\ref{lem:bend-discover} 
applied to $\beta^{\low}_{r(y)}$ and \blue{the
path $(\psi_1(x,v_1), \dots, \psi_k(v_{l-1},v_l))$ obtained from the 
cycle $C$ where we replace the last variable by a new variable $z$. 
Hence, if-clause in line 20 will apply for \red{$v=r(y)$}, i.e., 
$D$ will contain a closed walk. 
Let $D' := (\phi_1(x,u_1), \dots,\phi_k(u_{k-1},u_k))$ be the path obtained from $D$ by replacing repeated occurrences of variables by new variables. 
\red{Then it is clear from the algorithm that 
$\beta_{r(y)}^{\low} \wedge \phi_{D'}(r(y),z)$ implies a stronger bound for $z$ than $\beta(y) \wedge \phi_C(y,z)$.}
Lemma~\ref{lem:bend-detect} 
implies that $\phi_{D'}(y,y) \wedge \beta(y)$ is not handcuff consistent  or $\phi_1(x,u_1) \wedge \cdots \wedge \phi_k(u_{k-1},x) \wedge \psi_1(x,v_1) \wedge \cdots \wedge \psi_k(v_{l-1},x)$
is not handcuff consistent. The second case is impossible by our
assumption that $\Phi$ is handcuff consistent.
Hence, the first case applies. 
If $(\beta(y),(),\phi_{D'}(y,y))$ itself is a handcuff refutation then in the algorithm $\beta^{\low}_v$ or $\beta^{\high}_v$ gets replaced by false, and 
we obtain
a contradiction to the assumption that the algorithm did not return `No'. 
Otherwise, we have $D' = (F_1,\phi_i,F_2)$ and a literal $\gamma(w)$ of $\phi_i$ such that 
$(\beta,F_1,\gamma)$ or $(\gamma,F_2,\beta)$ \red{has a homomorphism to $\Psi$.} 
But this means that 
$\beta_w^{\high}$ or $\beta_w^{\low}$ would be in contradiction to $\gamma(w)$, 
contrary to the assumption that $\Psi$ is constructed from the non-redundant literals at the final state of the algorithm.}

To show that $\Psi$ is handcuff consistent, 
suppose otherwise that $\delta(y)$
is a bound which appears as a literal of a bend $\phi$ from $\Psi$ such that the instance $\Psi'$ obtained from $\Psi$ 
by replacing $\phi$ by $\delta$ has a handcuff refutation, \red{i.e., there is an unsatisfiable handcuff $(C,P,(\delta))$ with $|C|,|P| \leq |V|$ and a homomorphism $r$ to $\Psi$. 
Since $\Phi$ is satisfiable,}
we must have $C = (x \geq s)$. 
Consider the time point  
of the execution of the algorithm 
where all 
bounds $\beta_u^{\high}$ and $\beta_u^{\low}$
have been derived that witness that some literals in $\Phi$ are redundant at the final stage of the algorithm.
Then after $|V|$ more iterations of the second loop the algorithm would have derived a bound 
\red{$\beta=\beta^{\low}_{r(y)}$ or $\beta=\beta^{\high}_{r(y)}$} such that $\beta(y) \wedge \delta(y)$ is unsatisfiable, in contradiction to the assumption that $\Psi$ does not contain redundant literals. 
\qed 

\medskip 
{\bf Running time.} 
Let $n$ be the number of variables and 
$m$ be the number of constraints in $\Phi$.
We claim that the PROPAGATE 
performs $O(nm^2)$ arithmetic operations. 
The
outer loop of the algorithm is executed $O(m)$ many times. 
\blue{The execution of the loop `For $i=1,\dots,2 \cdot |V|$ do' in line 07 takes $O(nm)$ many steps.
The execution of the loop `For all $v \in V$ do'
takes $O(nm)$, too; so this matches the running time claimed by Hochbaum and Naor for the procedure that they credit to Aspvall and Shiloah).} 
The computation in the inner part of the loops
involves arithmetic operations. 
However, the representation sizes of these
numbers remain linear in the input size, so that we obtain a strongly polynomial bound on the running time of the algorithm. 
\qed

\subsection{Generalising Hochbaum-Naor}
\label{sect:hn-gen}
Using the procedure PROPAGATE from Figure~\ref{fig:alg-as-gen}, we can generalise the algorithm of Hochbaum-Naor for TVPI constraints from~\cite{HochbaumNaor} to bijunctive formulas.

\begin{definition}\label{def:gen-breakpoints}
Let $\Phi$ be a set of bends on the variables $x,y$. If there are two literals $a_1 x + b_1 y \circ_1 c_1$ and $a_2 x + b_2 y \circ_2 c_2$
in $\Phi$ such that
$a_1 x + b_1 y = c_1 \wedge a_2 x + b_2 y = c_2$ 
has exactly one solution $(u,v) \in {\mathbb Q}^2$, then $(u,v)$ is called a \emph{breakpoint} of $\Phi$. 
\end{definition}

\begin{figure}[h]
\fbox{
\begin{tabular}{l}
// Input: a bijunctive formula $\Phi$ with the variables $V$. \\
// Task: decide whether $\Phi$ is satisfiable. \\
\\
01: If $\Phi$ has only one variable, return whether
$\Phi$ is satisfiable or not \\ 
02: \hspace{.5cm} (this is straightforward to decide). \\
03: Otherwise, pick a variable $x$ of $\Phi$ (to be eliminated). \\
04: For $i \in \{1,\dots,|V|\}$, let $B_i$ be a list of all $u \in {\mathbb Q}$ such that  \\
05: \hspace{.5cm} there is a breakpoint $(u,v)$ for the constraints in $\Phi$ on the variables $x,x_i$, \\
06: \hspace{.5cm} or there is a bound $x \circ_1 u$ in $\Phi$ for
$\circ_1 \in \{<,>,\leq,\geq\}$. \\
07: Merge all the sequences $B_1,\dots,B_{|V|}$ into a sorted sequence $B = (b_1,\dots,b_k)$. \\
08: Perform a binary search on $B$ to find the largest $b_{\ell}$ \\
09: \hspace{.5cm} where the procedure PROPAGATE applied to $\Phi \wedge x \geq b_{\ell}$ returns `Yes'. \\
10: \hspace{.5cm} If such an $\ell$ does not exist, set $\ell :=0$ and $b_{\ell} = -\infty$. \\
11: \hspace{.5cm} If $\ell = k$, set $b_{\ell+1} := +\infty$. \\ 
12: Replace each bend $\phi$ in $\Phi$ that involves $x$ by the disjunct $\psi$ of $\phi$  \\
13: \hspace{.5cm}  such that $\psi \wedge x \geq b_{\ell} \wedge x < b_{\ell+1}$ is weakest (see correctness proof of algorithm).  \\
// After this step, all bends that involve $x$ are equivalent to \\
// bounds or two-variable linear inequalities. \\
14: Let $\Psi$ be obtained from $\Phi$ by removing all bends involving $x$. \\
15: For all conjuncts $\psi_1$ and $\psi_2$ of $\Phi$
involving $x$ that are strongest: \\
16: \hspace{.5cm} If $\psi_1$ is equivalent to $x \circ_1 a_1 y_1 + b_1$ for $\circ_1 \in \{>,\geq\}$ \\ 17: \hspace{.5cm} and $\psi_2$ is equivalent to $x \circ_2 a_2 y_2 + b_2$  and $\circ_2 \in \{<,\leq\}$: \\
18: \hspace{1cm}  Let $\circ$ be $<$ if $\circ_1$ or $\circ_2$ is strict, and let $\circ$ be $\leq$ otherwise. \\
19: \hspace{1cm} Add $a_1 y_1 + b_1 \circ a_2 y_2 + b_2$ to $\Psi$. \\
20: Return the result from recursively applying the algorithm to $\Psi$. 
\end{tabular}
}
\caption{An algorithm for satisfiability of a bijunctive formula over ${\mathbb Q}$.} 
\label{fig:alg-hn-gen}
\end{figure}

Our algorithm can be found in Figure~\ref{fig:alg-hn-gen}. The algorithm contains 
a step based on \emph{Fourier-Motzkin elimination} (lines 14-19). The idea of Fourier-Motzkin elimination is that if $\Phi$ is a system of linear inequalities on $n$ variables, and $x$ is a variable from $\Phi$, we can compute a system of linear inequalities which is equivalent to
$\exists x. \Phi$. In particular, this system is satisfiability-equivalent to $\Phi$. 
\begin{itemize}
\item If $\phi \in \Phi$ is equivalent
to $t_1 \leq x$ where $t_1$ is linear expression which does not involve $x$, 
and $\psi \in \Phi$ is equivalent to $x \leq t_2$ 
where $t_2$ is linear expression which does not involve $x$, 
 then $\Phi$ implies the inequality $C(\phi,\psi) := (t_1 \leq t_2)$. 
\item if $\phi_1,\dots,\phi_k$ are all the constraints in $\Phi$ that
yield a lower bound on $x$ when fixing all other variables, and $\psi_1,\dots,\psi_l$ are the constraints in $\Phi$ that similarly yield upper bounds, then 
$$\Phi \setminus \{\phi_1,\dots,\phi_k,\psi_1,\dots,\psi_l\} \cup \{C(\phi_i,\psi_j) \mid i \leq k, j \leq l\}$$ 
 is a system which is equivalent to $\exists x. \Phi$. 
 \end{itemize}
 Hence, we can test satisfiability of $\Phi$
 by eliminating all variables one-by-one and testing whether 
 the resulting formula is equivalent to $\top$ (true)
 or $\perp$ (false). 
 Note that if $\phi$ and $\psi$ are TVPI constraints, then $C(\phi,\psi)$ is a TVPI constraint as well. 
The problem with Fourier-Motzkin elimination is that the number of new constraints
$C(\phi,\psi)$ is quadratic in general, 
and since we have to repeat this step $n$ times
we only get an exponential worst-time upper bound on the space and time complexity of the algorithm. 
As in Hochbaum-Naor, we can avoid this 
by using the Fourier-Motzkin step only in situations
where we can guarantee that the number of new inequalities cannot repeatedly grow quadratically,
\blue{as explained below.} 

\medskip 
{\bf Correctness.} 
Let $\Phi$ be a bijunctive formula over the variables $V$. 
Arbitrarily choose $x \in V$. 
\blue{Let $B$ be the sorted list of the coordinates of all break-points as computed by the algorithm. 
Our algorithm in Figure~\ref{fig:alg-hn-gen}
performs an interval search to find the 
largest entry $b_{\ell}$ of $B$
 such that the procedure PROPAGATE from Figure~\ref{fig:alg-as-gen} answers 
 `Yes' on input `$\Phi \wedge x \geq b_\ell$'. 
 If $\ell = |B|$ then $b_{\ell+1} := \infty$. 
 Note that $\Phi$ has a solution if and only if
 $\Phi \wedge x \geq b_{\ell} \wedge x < x_{\ell+1}$ has a solution. 
 Also note that every bend in $\Phi$ that involves $x$ has a unique disjunct that is \emph{weakest} for all solutions satisfying $x \geq b_{\ell} \wedge x < x_{\ell+1}$ in the sense that replacing the bend by this disjunct does not change the set of solutions satisfying $x \geq b_{\ell} \wedge x < x_{\ell+1}$.} 

We can now eliminate $x$ 
using Fourier-Motzkin elimination. 
\blue{The resulting
instance $\Psi$ has a solution if and only if $\Phi$ has a solution.}
\qed 



\medskip 
{\bf Running time.}
We have to eliminate $n := |V|$ variables. 
Let $m$ be the number of conjuncts of $\Phi$. 
We use the procedure PROPAGATE; recall that the running time of PROPAGATE is in $O(n m^2)$. For the binary search, the
 algorithm PROPAGATE is called at most \blue{$O(\log m)$} times. 
 
\blue{Note that in the instance $\Phi \wedge x \geq b_{\ell} \wedge x \leq x_{\ell+1}$ 
for each variable 
$y$ distinct from $x$ there exist at most two
most restrictive two-variable linear inequalities 
on the variables $x,y$; otherwise, there would have to be a breakpoint whose $x$-value is between $b_{\ell}$ and $b_{\ell+1}$, contrary to the assumptions.}
Hence, in each step, the number of two-variable linear inequalities from the
Fourier-Motzkin elimination step is at most quadratic in $n$. Also note that 
all the numbers involved in the computation have a representation size that is linear in the input size, essentially because they are linear expressions in the numbers from the input.
\blue{
No new breakpoints are introduced in the computation.
}
So the overall running time is in 
\blue{$O(n (\log m) n m n^2) = O(n^4 m \log m)$}. \qed


\begin{theorem}\label{thm:main}
Let 
$\Phi$ be a finite conjunction of bounds, two-variable linear inequalities, and bends (where numbers are represented in binary). 
Then there is an algorithm that decides in polynomial time in the representation size of $\Phi$ whether $\Phi$ has a solution over ${\mathbb Q}$ or not. 
\end{theorem}

We mention that our algorithm is \emph{strongly polynomial}: it is polynomial in the Turing model and, additionally, it only performs a polynomial number of additions and multiplications of numbers in the input formulas to decide satisfiability. 
Corollary~\ref{cor:main} below states the consequence of Theorem~\ref{thm:main}
in the framework of constraint satisfaction problems; note that in this setting, the result
is independent from the encoding of the constraints in the input (since we have finitely many relations that are represented symbolically in the input). 

%

\begin{corollary}\label{cor:main} 
Let $R_1,\dots,R_{\ell}$ be semilinear relations on ${\mathbb Q}$ that are preserved by the median operation.
Then $\Csp({\mathbb Q};R_1,\dots,R_{\ell})$ can be solved in polynomial time. 
\end{corollary}
\begin{proof}
For a given instance $I$ of $\Csp({\mathbb Q};R_1,\dots,R_{\ell})$, we replace each atomic
formula $R_i(x_1,\dots,x_k)$ by $\phi(x_1,\dots,x_k)$ where $\phi$ is the definition of $R_i$ by a bijunctive definition  (Theorem~\ref{thm:syntax}). The resulting formula
is satisfiable over ${\mathbb Q}$ if and only if 
$I$ is satisfiable in $({\mathbb Q};R_1,\dots,R_{\ell})$. Hence, polynomial-time tractability follows
from Theorem~\ref{thm:main}. 
\end{proof}

\section{Maximal Tractability}
\label{sect:maximality}
Let $\Gamma$ be a relational structure
and let $\Delta$ be a structure obtained from
$\Gamma$ by dropping some of the relations
(i.e., $\Delta$ has a smaller signature).
In this case,  
$\Delta$ is called a \emph{reduct} of $\Gamma$ and $\Delta$ is called an \emph{expansion} of $\Gamma$; the expansion is called \emph{strict}
if $\Gamma$ has some relation not present in $\Delta$. A structure $\Gamma$ whose relations are semilinear is called \emph{maximally tractable} if 
\begin{itemize}
\item $\Csp(\Delta)$ is in P for every reduct of $\Delta$ with finite signature; and
\item every strict expansion $\Gamma'$ of $\Gamma$ has a finite-signature reduct $\Delta$
such that $\Csp(\Delta)$ is NP-hard. 
\end{itemize}

\begin{theorem}\label{thm:maximality}
The set $\Gamma$ of median-closed semilinear relations over ${\mathbb Q}$ is maximally tractable.
\end{theorem}
\begin{proof}[Sketch]
Polynomial-time tractability of $\Csp(\Delta)$ 
for finite subsets $\Delta$ of $\Gamma$ is Corollary~\ref{cor:main}. 
We show the second part of the statement by using the maximal tractability of the set of all 
median-closed constraints
on a finite linearly ordered domain $D$.
It is known (R. P\"oschel and L. Kalu\v{z}nin~\cite{KaluzninPoeschel}, Theorem 4.4.5) that the median operation on $D$ generates a \emph{minimal clone}.  
It follows from \blue{the discussion in~\cite{Max} (Section 2.3)}
that the set of all median-closed relations over
$D$ is maximally tractable. 
That is, for every relation $S \subseteq D^k$ that is not preserved by the median operation 
there are finitely many median-closed relations 
$S_1,\dots,S_{\ell}$ on $D$
such that $\Csp(D;S,S_1,\dots,S_{\ell})$ is NP-hard.  

Let $R \subseteq {\mathbb Q}^k$ be a relation which
is not median-closed. So there exist $k$-tuples
$t_1,t_2,t_3 \in R$ such that $\median(t_1,t_2,t_3) \notin R$. 
Let $D$ be the set of all entries of the tuples $t_1,t_2,t_3$. Note that the relation
$S := R \cap D^k$ is not preserved by the median
operation with respect to the linear order induced by the linear order of ${\mathbb Q}$ on $D$. 
Hence, there are finitely many median-closed relations 
$S_1,\dots,S_{\ell}$ on $D$
such that $\Csp(D;S,S_1,\dots,S_{\ell})$ is NP-hard. 

We now view the relations $D,S_1,\dots,S_{\ell}$
on $D$  
as relations over ${\mathbb Q}$; note that these
relations are preserved by the median operation on ${\mathbb Q}$. To show that 
$\Csp({\mathbb Q};R,D$, $S_1,\dots,S_{\ell})$ 
is NP-hard, it suffices to give a polynomial-time reduction from $\Csp(D;$ $S,S_1,\dots,S_{\ell})$ to $\Csp({\mathbb Q};R,D,S_1,\dots,S_{\ell})$. 
This follows from the fact that the relation $S$ can be defined as $$S(x_1,\dots,x_k) \Leftrightarrow R(x_1,\dots,x_k) \wedge D(x_1) \wedge \cdots \wedge D(x_k)$$
over the structure $({\mathbb Q};R,D,S_1,\dots,S_{\ell})$. 
\end{proof}

\section{Conclusion and Open Problems}
\label{sect:open}
We have identified a new polynomial-time tractable
class of constraints over the rational numbers that properly contains the class of TVPI constraints. Our class is maximally tractable
within the class of all semilinear relations.
\blue{A complexity dichotomy for all semilinear constraint languages would be a powerful
extension of the recently proved finite-domain dichotomy conjecture~\cite{FederVardi,BulatovFVConjecture,ZhukFVConjecture}. } \blue{See~\cite{Jonsson2016912} for other results towards this ambitious research goal.}

We mention that there are two other candidates of maximally tractable classes of semilinear relations,
namely the class of semilinear relations preserved by $\max$, and the class preserved by $\min$; the second item in the definition of maximal tractability has been shown in~\cite{BodirskyMaminoTropical}, but the first item, i.e., polynomial-time tractability, is open.
\blue{Indeed, a polynomial-time algorithm for $\max$-closed semilinear constraints would in particular give a polynomial-time algorithm for the max-atom problem and mean payoff games; see the discussion in~\cite{BodirskyMaminoTropical}. }

Another interesting concrete example of a constraint satisfaction problem with semilinear constraints over the rationals 
that we conjecture to be in P is 
$$\Csp({\mathbb Q};X,S)$$
\begin{align*}
\text{ where} \quad  X & := \{(x,y,z) \in {\mathbb Q}^3 \mid (x=y<z) \vee (y=z<x) \vee (z=x<y)\} \\
\text{and} \quad S & := \{(x,y) \in {\mathbb Q}^2 \mid x=y+1\} .
\end{align*}
Note that $\Csp({\mathbb Q};X)$ and $\Csp({\mathbb Q};S)$ are known to be in P, but it is unclear how to generalise the algorithm given in~\cite{tcsps-journal} for $\Csp({\mathbb Q};X)$ to deal with constraints of the form $x = y + 1$.

\bibliographystyle{abbrv}
\bibliography{../../global.bib}

\appendix

\section{Proof of Theorem~\ref{thm:shostak}}
\label{sect:proof-Shostak}
We show Shostak's theorem in two steps.

\begin{lemma}\label{lem:shostak} 
An TVPI instance is unsatisfiable if and only if it has a handcuff refutation. 
\end{lemma}
\begin{proof}
We show inductively that for every $i \leq n$
there exists a bijunctive formula $\Phi_i$ without handcuff refutation 
and 
a map $s_{i} \colon S_{i} \to {\mathbb Q}$ such that $\Phi_i$ implies
 \begin{align}
 \Phi \wedge \bigwedge_{x \in S_{i}} (x \geq s_{i}(x) \wedge x \leq s_{i}(x)).
 \label{eq:key}
 \end{align}
 If we succeed defining such a map for $i = n$, then $s_n$ is a solution to $\Phi$ and we are done. 
 The map $s_0$ has 
empty domain and $\Phi_0 := \Phi$ satisfies the statement by assumption. So suppose that we have already defined $s_{i-1}$ for $i \leq n$, and that we want to define $s_i$. 
For $j=1$ and $j=2$, let $P_j$ be a path from $z_j$ to $y$ with a variable set $U_j$ of size at most $n$ and $C_j$ a
cycle starting and ending in $z_j$ with a variable set $V_j$
of size at most $n$ such that $U_j \cap V_j = \{y\}$ 
and $r_j \colon U_j \cup V_j \to V$ a homomorphism from $P_jC_j$ to $\Phi_{i-1}$ 
such that $r_j(y) = x_i$,
the lower bound $\alpha_1(y)$ 
implied by $\phi_{P_1}(z_1,y) \wedge \phi_{C_1}(y)$ is strongest possible, and
the upper bound $\alpha_2(y)$ 
implied by $\phi_{P_2}(z_2,y) \wedge \phi_{C_2}(y)$ 
is strongest possible.

\medskip 
{\bf Claim 1.} $\alpha_1(y) \wedge \alpha_2(y)$
has a solution $c \in {\mathbb Q}$. 

\red{Otherwise, let $P'_2,C_2'$ be obtained 
from $P_2,C_2$ by 
renaming the variables so that $C_1,P_1$
and $P_2',C_2'$ only share the variable $y$. 
Then there is a homomorphism from the unsatisfiable handcuff $(C_1,P_1,P'_2,C'_2)$ to
$\Phi_{i-1}$, in contradiction to the assumptions.}

\medskip 
We now define $s_i$ as the extension of $s_{i-1}$ where $s_i(x_i) = c$. 

\medskip
{\bf Claim 2.} 
$\Phi_i := \Phi_{i-1} \wedge x_i \leq c \wedge x_i \geq c$ 
 does not have a handcuff refutation: since $\Phi_{i-1}$ does not contain
such a refutation, any unsatisfiable handcuff 
$(C,P,D)$ with homomorphism to $\Phi_i$ must be of the form $((x_i \leq c),P,D)$ or of the form
$(D,P,(x_i \leq c))$. 
But then $\phi_{P}(x_i,y) \wedge \phi_{D}(y)$
\red{and therefore $\Phi_{i-1}$}
would imply a stronger lower bound for $x_i$ 
than $\alpha_1(x_i)$, a contradiction. The second case is similar. 
Note that $\Phi_i$ implies~(\ref{eq:key})
and that $\Phi_i$ and $s_i$ therefore satisfy the requirements, concluding the induction. 
\end{proof}

Theorem~\ref{thm:shostak} is an immediate
consequence of Lemma~\ref{lem:shostak}
and the following lemma. 

\begin{lemma}
If $\Phi$ is an TVPI instance with a handcuff refutation. Then $\Phi$ contains a path $P$
from $x_0$ to $x_k$ and two cycles with residue inequality $\alpha(x_0)$ and $\beta(x_1)$ 
such that $\alpha(x_0) \wedge \res_P(x_0,x_k) \wedge \beta(x_k)$ is unsatisfiable.
\end{lemma}
\begin{proof}
Let $h$ be a homomorphism from the unsatisfiable handcuff $(C,P,D)$ to $\Phi$. 
Pick $h$ and $(C,P,D)$ such that 
$|CPD|$ is minimal. 
If the image of 
$C$ and of $D$ under $h$ in $\Phi$  are cycles 
and the image of $P$ under $h$ in $\Phi_{i-1}$ is a path then we are done. 
Otherwise, if the image of $P$ under $h$ is not a path there must 
be variables $z_1,z_2$ in $P$ such that $h(z_1)=h(z_2)$. 
Write $P$ as $Q_1Q_2Q_3$ where
$Q_2$ is a path from $z_1$ to $z_2$. 
Let $E$ be the cycle obtained from $Q_2$ by replacing both $z_1$ and $z_2$ by a new variable $z$. Note that
$\res_{CQ_1}(z) \wedge \res_{E}(z) \wedge \res_{Q_3D}(z)$ is unsatisfiable.
We distinguish the following cases.
\begin{itemize}
\item $\res_{CQ_1}(z) \wedge \res_{Q_3D}(z)$  is unsatisfiable. 
\item $\res_{C_1Q_1}(z) \wedge 
\res_{E}(z)$ is unsatisfiable.  
\item $\res_{E}(z) \wedge \res_{Q_2C'_2}(z)$  is unsatisfiable.  
\end{itemize}
 In all three cases, we have found a smaller handcuff refutation of $\Phi$, a contradiction to the minimal choice of $(C,P,Q)$. 
 
Now suppose $C$ is not a cycle;
 the case that $D$ is not a cycle is analogous. There must be variables $z_1,z_2$ that appear in $C$ such that $h(z_1)=h(z_2)$. 
Write $C$ as $C_1QC_2$ where
$Q$ is a path from $z_1$ to $z_2$. 
Let $E$ be the cycle obtained from $Q$ by replacing both $z_1$ and $z_2$ by a new variable $z$. Note that
$\res_{DPC_1}(z) \wedge \res_{E}(z) \wedge \res_{DPC_2}(z)$ is unsatisfiable. 
Again we distinguish three cases and obtain
a shorter handcuff refutation in each case, contradiction. 
\end{proof} 

\end{document}

%% file: bijunctive.tex
\def\Q{\mathbb Q}

\vskip\baselineskip
\noindent{\bf Claim.}
Let $V\subseteq \Q^2$ be a semilinear set such
$V$ contains all $(x',y') \in \Q^2$ that 
there exists $(x,y)\in V$ with $x'\le x$ and $y'\le y$. 
Then $V$ has a bijunctive definition.
\vskip\baselineskip

If $V \in \{\emptyset,\Q^2\}$ the statement is
trivial; so assume that this is not the case. We associate to $V$ the following functions
\begin{align*}
f_V \colon \Q &\to \Q \\
t &\mapsto \sup\{r \mid (r+t,r-t)\in V\}\\
g_V \colon \Q &\to \{0,1\}\\
t &\mapsto \begin{cases}
0 &\text{if $(f_V(t)+t,\,f_V(t)-t) \notin V$} \\
1 &\text{if $(f_V(t)+t,\,f_V(t)-t) \in V$}
\end{cases}
\end{align*}
For an illustration, see Figure~\ref{fig:syntax}.

\begin{figure}[h]
  \begin{center}
     \includegraphics[scale=0.8]{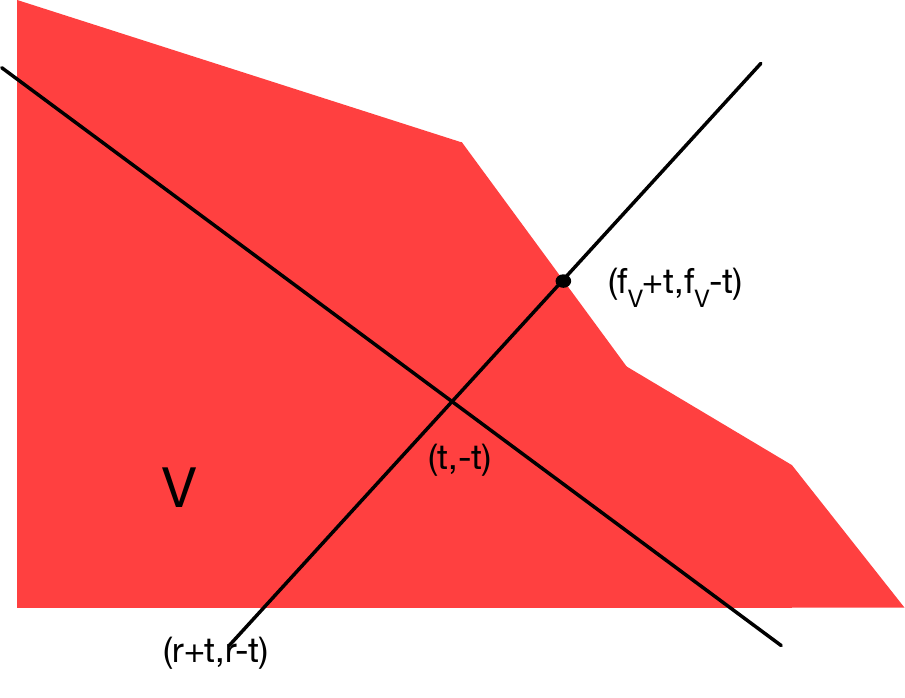}
  \end{center}
\caption{An illustration for the proof of the implication $(2) \Rightarrow (3)$ in Theorem~\ref{thm:syntax}.}
  \label{fig:syntax}
\end{figure}

Observe that the supremum takes values in~$\Q\cup\{-\infty,\infty\}$ by
quantifier elimination. Moreover the values $-\infty,\infty$ can be
excluded. In fact, 
for every $t \in {\mathbb Q}$  
 the set of~$r \in {\mathbb Q}$ such that~$(r+t,r-t)\in V$ is obviously
either empty, or~$\Q$, or a left half-line, and we need to exclude the
first two cases. This set can not be empty because, given any point~$(x,y)$
in~$V$, we can ensure $(r+t,r-t)\in V$ by choosing~$r$ in such a way that
$r+t\le x$ and~$r-t\le y$. Similarly it can not be~$\Q$ because, given a
point~$(x,y)$ not in~$V$, it suffices to choose~$r$ so that $r+t\ge x$
and~$r-t\ge y$. This proves that $f_V$ and~$g_V$ are well-defined
semilinear functions.

By the above, the function mapping $t\in\Q$ to the pair~$(f_V(t),g_V(t))$
is semilinear (i.e., the graph of this function is a semilinear relation), hence, again by quantifier elimination, we have
$-\infty=k_0<k_2 < \dotsb < k_n=\infty$ with $k_1\dotsc k_{n-1}\in\Q$ such that
for all~$0\le i < n$
\begin{align*}
f_V|_{]k_i,k_{i+1}[}(t) &= a_i t + b_i \\
g_V|_{]k_i,k_{i+1}[}(t) &= c_i
\end{align*}
for appropriate $a_i,b_i  \in \mathbb Q$ and~$c_i\in\{0,1\}$. 
We define the formulas $\phi_0, \phi_1,\dots,\phi_{n-1}$ as follows: for all $i \in \{1,\dots,n-2\}$ set 
\begin{align*}
\phi_i \; & := \;
y < f_V(k_{i})-k_{i} \; \vee \; (1-a_i) x + (1+a_i) y < 2b_i  \; \vee \; y < f_V(k_{i+1})-k_{i+1}
\end{align*}
if $g_V(\frac{k_i+k_{i+1}}{2})=0$, and otherwise
$$y < f_V(k_{i})-k_{i} \; \vee \; 
(1-a_i) x + (1+a_i) y \leq 2b_i 
 \; \vee \; y < f_V(k_{i+1})-k_{i+1}.$$
Moreover, 
\begin{align*}
\phi_0 & := \begin{cases}
y < a_0 x + b_0 \; \vee \; y < f_V(k_1)-k_1 & \text{ if } g_V(k_1-1)=0\\
y \leq a_0 x + b_0 \; \vee \; y < f_V(k_1)-k_1 & \text{ otherwise } 
\end{cases} \\
\phi_{n-1} & := \begin{cases}
y < f_V(k_{n-1})-k_{n-1} \; \vee \; y < a_n x + b_n  & \text{ if } g_V(k_{n-1}+1)=0\\
y < f_V(k_{n-1})-k_{1-1} \; \vee \; y \leq a_0 x + b_0 & \text{ otherwise.} 
\end{cases}
\end{align*}
For~$0<i<n$, we define
\begin{align*}
\psi_i \;:=\;
\begin{cases}
x < f_V(k_i)+k_i \; \vee \;
y < f_V(k_i)-k_i & \text{ if } g_V(k_i)=0 \\
x \leq f_V(k_i)+k_i \; \vee \;
y \leq f_V(k_i)-k_i & \text{ otherwise.}
\end{cases}
\end{align*}
Observe that the formulas $\psi_i$
are bends. 
We claim that the following is a bijunctive definition of~$V$.
\[
\Phi \;:=\; \bigwedge_{0\le i <n} \phi_i \, \wedge\, \bigwedge_{0<i<n}
\psi_i 
\]
To begin with, we need to show that the formulas~$\phi_i$ are actually
bends, namely that $-1\le a_i\le 1$. Assume $a_i < -1$, the other case is symmetric. Fix
$t$ such that $k_i<t<k_{i+1}$ and let $\epsilon$ denote a positive
rational chosen small enough that~$k_i<t-\epsilon<t+\epsilon<k_{i+1}$.
Consider the points
\begin{align*}
p_1 &:= \left(
(a_i+1) (t-\epsilon) + b_i + \frac{a_i+1}{2} \epsilon
,\,
(a_i-1) (t-\epsilon) + b_i + \frac{a_i+1}{2} \epsilon
\right)\\
p_2 &:= \left(
(a_i+1) (t+\epsilon) + b_i - \frac{a_i+1}{2} \epsilon
,\,
(a_i-1) (t+\epsilon) + b_i - \frac{a_i+1}{2} \epsilon
\right) 
\end{align*}
Clearly $p_1\in V$, because
\begin{align*}
(a_i+1) (t-\epsilon) + b_i + \frac{a_i+1}{2} \epsilon =
f_V(t-\epsilon) + (t-\epsilon) + \frac{a_i+1}{2} \epsilon <
f_V(t-\epsilon) + (t-\epsilon)
\\
(a_i-1) (t-\epsilon) + b_i + \frac{a_i+1}{2} \epsilon =
f_V(t-\epsilon) - (t-\epsilon) + \frac{a_i+1}{2} \epsilon <
f_V(t-\epsilon) - (t-\epsilon)
\end{align*}
and similarly $p_2\notin V$. Yet the coordinates of~$p_2$ are
componentwise strictly smaller than those of~$p_1$, contradicting the hypothesis
on~$V$.

It remains to show that $\Phi$ defines~$V$. By
inspection of the formula, $\Phi$ is equivalent to
\[
\Psi \;:=\;
\forall\, t\;: \; \begin{cases}
g_V(t)=0\, \rightarrow\, (\,x<f_V(t)+t \,\,\vee\,\, y<f_V(t)-t\,)\\
\quad\quad\quad\quad\wedge\\
g_V(t)=1\, \rightarrow\, (\,x\le f_V(t)+t \,\,\vee\,\, y\le f_V(t)-t\,)
\end{cases}
\]
Now, if $(x,y)\notin V$, considering $t=\frac{x-y}{2}$, we have $g_V(t) =0$, and since $f_V(t)+t \leq x$ and $f_V(t)-t \leq y$ by the definition of $f_V$  
it follows that
$\Psi$ must fail. Conversely, assume~$\neg\Psi(x,y)$ and let~$t$ witness
this. If $g_V(t)=0$, then $(f_V(t)+t,\,f_V(t)-t)\notin V$ hence, because
$f_V(t)+t\le x$, $f_V(t)-t\le y$, we get~$(x,y)\notin V$.
If $g_V(t)=1$, then
$f_V(t)+t< x$ and $f_V(t)-t< y$, so choosing a positive~$\epsilon$ such
that $f_V(t)+\epsilon+t< x$ and $f_V(t)+\epsilon-t< y$, we have
$(f_V(t)+\epsilon+t,\,f_V(t)+\epsilon-t)\notin V$ and~$(x,y)\notin V$ as
above.